\documentclass[12pt, onecolumn]{IEEEtran}
\IEEEoverridecommandlockouts
\usepackage[utf8]{inputenc} 
\usepackage[T1]{fontenc}
\usepackage{url}
\usepackage{ifthen}
\usepackage{cite}
\usepackage[cmex10]{amsmath} 
\usepackage{amssymb}
\usepackage{graphicx}
\usepackage{epstopdf}
\usepackage{epsfig}
\usepackage{multirow}
\usepackage{amsthm}
\usepackage{comment}
\usepackage{caption}
\usepackage{enumitem}
\usepackage{amsmath}
\usepackage{color}
\usepackage{bbm}
\usepackage{algorithm}
\usepackage{algpseudocode}
\usepackage{csquotes}
\usepackage{subfiles}
\usepackage{caption, multirow, makecell}
\usepackage{array}
\usepackage{caption}
\usepackage{subcaption}
\usepackage{afterpage}
\usepackage{dblfloatfix}
\usepackage{hyperref}
\usepackage{booktabs}
\usepackage{longtable}
\usepackage{blkarray}
\usepackage{adjustbox}
\newtheorem{thm}{Theorem}
\newtheorem{lem}{Lemma}
\newtheorem{prop}{Proposition}
\newtheorem{cor}{Corollary}
\newtheorem{example}{Example}

\newtheorem{defn}{Definition}

\newtheorem{cons}{Construction}

\newtheorem{rem}{Remark}

\allowdisplaybreaks[1]

\def\BibTeX{{\rm B\kern-.05em{\sc i\kern-.025em b}\kern-.08em
    T\kern-.1667em\lower.7ex\hbox{E}\kern-.125emX}}

\begin{document}

\title{Function-Correcting Lee-distance Codes for Symbol-Pair Read Channels}
\author{\IEEEauthorblockN{Hareesh K. and B. Sundar Rajan} \\
	\IEEEauthorblockA{\textit{Department of Electrical Communication Engineering,} \\
		\textit{Indian Institute of Science,}
		Bengaluru, India \\
		\{hareeshk, bsrajan\}@iisc.ac.in}
}

\maketitle

\begin{abstract}
\emph{Function-correcting codes} (FCCs) are designed to protect a specified function evaluation of messages at a higher level than the level of protection for messages against errors, while reducing the redundancy required for reliable communication. Among the various channel models considered for FCCs, the \emph{symbol-pair read channel} is motivated by high-density data storage systems in which symbols are read in overlapping pairs rather than individually. \emph{Phase-shift keying} (PSK) modulation is well-suited to such storage systems due to its bandwidth efficiency and robustness to noise. While FCCs for symbol-pair read channels have previously been studied under the Hamming metric, the \emph{Lee metric} provides a more appropriate error model for $q$-ary PSK modulation and is therefore a natural choice for $q$-ary symbol-pair read channels. We first generalize the symbol-pair Lee-distance, previously defined only over $\mathbb{Z}_4$, to $\mathbb{Z}_q$, where $q \ge 2$. Based on this generalized metric for $q=2^m$, where $m\ge 1$, we initiate the study of function-correcting Lee-distance codes for symbol-pair read channels and introduce the notion of \emph{function-correcting symbol-pair Lee-distance codes} (FCSPLCs). We investigate the redundancy requirements of FCSPLCs. To this end, we introduce \emph{irregular-pair Lee-distance codes} and derive lower and upper bounds on the optimal redundancy in terms of the shortest length of such codes. We derive \emph{Plotkin-type} and \emph{Gilbert--Varshamov-type} bounds on the \emph{optimal redundancy} through lower and upper bounds on the shortest length of irregular-pair Lee-distance codes. By specializing to bijective functions, these results yield corresponding \emph{Plotkin-type} and \emph{Gilbert--Varshamov-type} bounds for classical Lee metric codes for symbol-pair read channels over $\mathbb{Z}_q$. To the best of our knowledge, these are the first such bounds established for the Lee metric symbol-pair setting. We then specialize the general FCSPLC framework to specific functions, including \emph{pair-locally bounded functions}, the \emph{Pair-Lee weight function}, and the \emph{Pair-Lee weight distribution function}. For these functions, we present \emph{explicit constructions} of FCSPLCs and derive corresponding bounds on the optimal redundancy. Finally, we study FCSPLCs associated with \emph{linear functions} and derive a Plotkin-type lower bound on their optimal redundancy.
\end{abstract}


\section{Introduction}

Conventional communication systems typically employ error-correcting codes (ECCs) to ensure reliable recovery of messages transmitted over noisy channels. In many scenarios, however, the objective of communication is not to recover the entire message but rather to determine the value of a specific function of the underlying data. Such scenarios arise in archival storage, distributed computing, and IoT sensor networks, where certain applications require only the function value rather than the complete message. Although recovery of the entire message enables computation of the desired function value, this approach can be inefficient when the message space is large while the function output space is relatively small. Motivated by this observation, Lenz \textit{et al.} \cite{Lenz} introduced the notion of \emph{function-correcting codes} (FCCs). FCCs ensure the reliable recovery of prescribed function values in the presence of channel errors while potentially requiring significantly less redundancy than conventional ECCs. Consequently, a central problem in the study of FCCs is to determine the minimum redundancy, referred to as the \emph{optimal redundancy}, required for reliable recovery of the desired function value. To address this problem, the authors introduced the notion of \emph{irregular-distance codes} and established a fundamental connection between the shortest lengths of such codes and the optimal redundancy of FCCs. Leveraging this connection, they derived general upper and lower bounds on the optimal redundancy. These results were subsequently specialized to several important function classes, including the Hamming weight and Hamming weight distribution functions, for which general FCC constructions were also presented.
   
   \subsection{Motivation}
    The study of FCCs under diverse channel models is important both for expanding their range of applications and for gaining a deeper understanding of their fundamental limits. One channel model of considerable practical relevance is the symbol-pair read channel \cite{Cassuto}, which arises in high-density data storage systems where physical limitations of the read process cause symbols to be read in overlapping pairs rather than individually. Since its introduction, this model has attracted significant attention in coding theory and has led to the development of symbol-pair codes capable of correcting pair-read errors. Motivated by these developments, FCCs for symbol-pair read channels were subsequently investigated under the Hamming metric \cite{XLC}. In many practical communication and storage systems, however, errors manifest as small perturbations of transmitted symbols, making the magnitude of symbol deviations an important consideration. This is particularly true for systems employing $q$-ary phase-shift keying (PSK) modulation, which is well suited to high-density storage applications due to its bandwidth efficiency and robustness to noise \cite{PAA}. In such settings, the Lee metric provides a more realistic measure of symbol discrepancy than the Hamming metric, as it captures the cyclic structure of the alphabet $\mathbb Z_q$ and more accurately reflects the error characteristics of $q$-ary PSK constellations. The Lee metric has therefore been widely used in coding theory for non-binary communication and storage systems where errors are naturally modeled as small-magnitude symbol perturbations \cite{CW,KN,AB}. For these reasons, the Lee metric is a natural choice for $q$-ary symbol-pair read channels, where $q=2^m$ and $m\ge1$. Although Lee metric codes for symbol-pair read channels have been studied over $\mathbb Z_4$ \cite{PAA}, a corresponding Lee-metric framework for general alphabets $\mathbb Z_q$, where $q \ge 2$, has not been investigated. Moreover, while FCCs have been studied separately under the Lee metric \cite{GA,HRS} and for symbol-pair read channels under the Hamming metric \cite{XLC}, function-correcting codes for symbol-pair read channels under the Lee metric remains unexplored. This motivates the development of a symbol-pair Lee metric over $\mathbb Z_q$ and the study of FCCs for $q$-ary symbol-pair read channels under this metric.
  \subsection{Related Works}
  Following the introduction of FCCs in \cite{Lenz}, the field has attracted considerable attention, leading to improved bounds on optimal redundancy, explicit code constructions, and several extensions and generalizations of the FCC framework. In \cite{PR}, the authors derived a lower bound on the optimal redundancy of FCCs using the notion of a function-dependent graph and obtained a simplified bound for linear functions. An upper bound for FCCs over binary fields, within a logarithmic factor of the known lower bound in \cite{Lenz}, was established in \cite{LS}. Improved lower and upper bounds on the optimal redundancy for the Hamming weight function were presented in \cite{ZXZG}. FCCs for a class of functions that assume only a limited number of values within a given Hamming ball, referred to as locally bounded functions, were studied in \cite{CSFH}. The FCC framework was further generalized in several directions. In \cite{CSFH2}, the authors introduced \emph{function-correcting partition codes}, which are defined directly on partitions of the message space rather than on specific target functions. Another generalization, termed \emph{FCCs with data protection}, was proposed in \cite{CSFH3}, where both the message and its function value are protected simultaneously. FCCs with data protection for the Hamming code membership function were subsequently investigated in \cite{SAARB,AAS}. Recently, data protection in function-correcting symbol-pair codes was addressed in \cite{ASAK}.
  
  While most of the aforementioned works study FCCs under the Hamming metric, the FCC framework has also been extended to several other channel models and distance measures. These include symbol-pair read channels \cite{XLC}, $b$-symbol read channels \cite{AAY,SR}, channels matched to homogeneous distance \cite{LL}, Lee-distance channels \cite{GA,HRS}, and insertion-deletion channels \cite{AA}. Among these, symbol-pair read channels have received particular attention. In \cite{XLC}, the authors initiated the study of FCCs for symbol-pair read channels by introducing function-correcting symbol-pair codes (FCSPCs). They investigated the optimal redundancy of FCSPCs through the framework of irregular-pair-distance codes and presented general constructions for such codes. Subsequently, FCCs for locally $(\lambda,\rho,b)$-functions over $b$-symbol read channels were investigated in \cite{FCBSC}. Among these extensions, FCCs for symbol-pair read channels and FCCs for Lee-distance channels have been studied independently. However, to the best of our knowledge, FCCs for symbol-pair read channels under the Lee metric have not been investigated. The present work addresses this gap by introducing and studying function-correcting Lee-distance codes for symbol-pair read channels.
  
  \subsection{Contributions}
  
  In this paper, we generalize the notion of symbol-pair Lee-distance from $\mathbb{Z}_4$ to $\mathbb{Z}_q$, where $q \ge 2$, thereby establishing a Lee-metric framework for $q$-ary symbol-pair read channels. Using  this framework for $q=2^m$, where $m\ge 1$, we introduce the notion of \emph{function-correcting symbol-pair Lee-distance codes} (FCSPLCs) and investigate their redundancy requirements. In particular, we derive lower and upper bounds on the optimal redundancy in terms of the shortest length of irregular-pair Lee-distance codes and specialize the FCSPLC framework to several important functions. We also present explicit constructions of FCSPLCs for these functions and study FCSPLCs associated with linear functions.
  
  The main contributions of this paper are summarized as follows:
  
  \begin{itemize}
  	\item We generalize the symbol-pair Lee-distance, previously defined only over $\mathbb{Z}_4$, to $\mathbb{Z}_q$,  where $q \ge 2$, prove that it defines a metric, and derive its relationship with the classical Lee-distance. This provides a Lee-metric framework for $q$-ary symbol-pair read channels.
  	
  	\item We develop a function-correcting coding framework for symbol-pair read channels over $\mathbb{Z}_q$, where $q=2^m$ and $m\ge 1$ under the Lee metric by introducing FCSPLCs and establishing their connection to irregular-pair Lee-distance codes.
  	
    \item Plotkin-type and Gilbert--Varshamov-type bounds on the optimal redundancy of FCSPLCs are derived. By specializing these results to bijective functions, we obtain corresponding Plotkin-type and Gilbert--Varshamov-type bounds for classical Lee metric codes for symbol-pair read channels over $\mathbb{Z}_q$. To the best of our knowledge, these constitute the first such bounds in the Lee metric symbol-pair setting.
  	
  	\item The proposed general FCSPLC framework is specialized to pair-locally bounded functions, the pair-Lee weight function, and the pair-Lee weight distribution function. Explicit FCSPLCs constructions are presented for these functions, together with corresponding upper bounds on the optimal redundancy.
  	
  	\item Finally, we study FCSPLCs associated with linear functions and derive a Plotkin-type lower bound on their optimal redundancy.
  \end{itemize}
 
   \subsection{Organization}
   
   The rest of the paper is organized as follows. Section~\ref{preliminaries} reviews the basic definitions and concepts related to Lee-distance codes over $\mathbb{Z}_q$ and symbol-pair Lee-distance codes over $\mathbb{Z}_4$. Section~\ref{section3} generalizes the notion of symbol-pair Lee distance from $\mathbb{Z}_4$ to $\mathbb{Z}_q$, where $q \ge 2$, and investigates its fundamental properties. Section~\ref{section4} introduces function-correcting symbol-pair Lee-distance codes (FCSPLCs), establishes general lower and upper bounds on their optimal redundancy through their connection to irregular-pair Lee-distance codes, and subsequently derives simplified bounds, a Plotkin-type lower bound, and a Gilbert--Varshamov-type upper bound. Section~\ref{section5} presents explicit constructions of FCSPLCs for pair-locally bounded functions, the pair-Lee weight function, and the pair-Lee weight distribution function, together with corresponding upper bounds on the optimal redundancy. Section~\ref{section6} derives a Plotkin-type lower bound on the optimal redundancy of FCSPLCs for linear functions. Finally, Section~\ref{conclusion} concludes the paper and outlines directions for future research.
   
   \noindent \textit{Notations:}
   $\mathbb{N}_0$ denotes the set of non-negative integers and $\mathbb{N}_{0}^{M\times M}$ represents the set of all  $M\times M$ matrices with non-negative integer entries. For a matrix $\boldsymbol{D}\in\mathbb{N}_{0}^{M\times M}$,  $[\boldsymbol{D}]_{ij}$ denotes its $(i,j)$th entry. For any two vectors $\boldsymbol{\mathbf{x},\mathbf{y}}\in \mathbb{Z}_q^n$, let $d_H(\boldsymbol{\mathbf{x},\mathbf{y}})$, $d_L(\boldsymbol{\mathbf{x},\mathbf{y}})$, and $d_p(\boldsymbol{\mathbf{x},\mathbf{y}})$ denote the Hamming distance, Lee distance, and symbol-pair Lee-distance between  $\mathbf{x}$ and $\mathbf{y}$, respectively. For any two real numbers $a$ and $b$, $\max\{a,b\}$ and $\min\{a,b\}$ denote the maximum and minimum of $a$ and $b$, respectively. For any integer $M$, \( [M]^+ \triangleq \max\{M,0\}\) and \( [M] \triangleq \{1,2,\ldots,M \}\). For positive integers $a$ and $b$, the shifted modulo operation is defined by $a \operatorname{smod} b=\big((a-1) \bmod b\big)+1$. Also $\{a\}^k \triangleq (\underbrace{a, a, \ldots, a}_{k\text{ times}})$, where the element $a$ is repeated $k$ times. 
   
\section{Preliminaries} \label{preliminaries}
In this section, we review the basic definitions and concepts related to the Lee-distance codes over $\mathbb Z_q$, and symbol-pair Lee-distance codes over $\mathbb Z_4$.
    \subsection{Lee-Distance Codes}
    We first recall the definitions of the Lee weight and Lee-distance over $\mathbb Z_q$, where $q \ge 2$.
    \begin{defn}[\!\!\cite{CW}]
    	The Lee weight of a symbol $x_{i} \in \mathbb{Z}_{q}$, i.e., $w_{L}(x_{i}) = \min(x_{i}, q - x_{i})$. The Lee weight of a vector $\boldsymbol{x} = (x_{1}, x_{2},\ldots, x_{n}) \in \mathbb{Z}_{q}^n$ is the sum of the  Lee weights of its components, i.e. $w_{L}(\boldsymbol{x}) = \sum_{i=1}^n w_{L}(x_{i})$.
    \end{defn}
    
    \begin{defn}[\!\!\cite{CW}]
    	 The Lee-distance between any two symbols  $x_{i}$ and $x_{j} \in \mathbb{Z}_{q}$, i.e., $d_{L}(x_{i},x_{j}) = \min(|x_{i} - x_{j}|, q - |x_{i} - x_{j}|)$ and the Lee-distance between two vectors $\boldsymbol{x} = (x_{1}, x_{2},\ldots, x_{n})$ and  $\boldsymbol{y} = (y_{1}, y_{2},\ldots, y_{n}) \in \mathbb{Z}_{q}^n$, i.e., $d_{L}(\boldsymbol{x},\boldsymbol{y}) = \sum_{i=1}^n \min(|x_{i} - y_{i}|, q - |x_{i} - y_{i}|)$. It can also be written as, $d_{L}(\boldsymbol{x},\boldsymbol{y}) = \sum_{i=1}^{n} w_{L} ((x_i - y_i) \bmod q)$.
    \end{defn}
    Let $\boldsymbol{x} = (x_{1}, x_{2},\ldots, x_{n}) \in \mathbb{Z}_{q}^n$, and $\boldsymbol{y} = (y_{1}, y_{2},\ldots, y_{n})$ is the result of at most \textit{t} errors from $\boldsymbol{x}$ if ${d_{L}(\boldsymbol{x}, \boldsymbol{y})} \le$ \textit{t}.
    The following lemma provides bounds on the Lee-distance in terms of the Hamming distance for vectors over $\mathbb{Z}_{q}^n$.
    \begin{lem}[\!\!\cite{CW}]
    	Let \( \boldsymbol{x}, \boldsymbol{y} \in \mathbb{Z}_q^n \), where \( q \geq 2 \). Then the following bounds hold:
    	\[
    	d_H(\boldsymbol{x}, \boldsymbol{y}) \leq d_L(\boldsymbol{x}, \boldsymbol{y}) \leq \left\lfloor \frac{q}{2} \right\rfloor d_H(\boldsymbol{x}, \boldsymbol{y}). 
    	\] 
    \end{lem}
    The following lemma characterizes the error-correcting capability of Lee-distance codes in terms of their minimum Lee-distance.
    \begin{lem}[\!\!\cite{CW}]
    	A Lee-distance code with minimum Lee-distance \( d_L \) can correct up to \( t \) errors if and only if
    	\(
    	d_L \geq 2t + 1.
    	\)
    \end{lem}
    \subsection{Symbol-Pair Lee-Distance Codes over $\mathbb{Z}_4$}
    In this subsection, we recall the definitions of the symbol-pair Lee weight and symbol-pair Lee-distance over $\mathbb{Z}_4$, and review the relationship between the symbol-pair Lee-distance and the classical Lee-distance.
    
    Let $\mathbf{x}=(x_1, x_2, \ldots, x_n)$ be a vector in $\mathbb{Z}_4^n$. Then the corresponding symbol-pair vector of $\mathbf{x}$ is defined as $\pi(\mathbf{x})=\big((x_1, x_2), (x_2, x_3), \ldots, (x_{n-1}, x_n), (x_n, x_1)\big)$, where $x_{n+1}=x_1$.
    \begin{defn}[\!\!\cite{PAA}]
    	Let $\mathbf{x}=(x_1, x_2, \ldots, x_n) \in \mathbb{Z}_4^n$. The symbol-pair Lee weight of $\mathbf{x}$ is defined as $\mathrm{w}_p(\mathbf{x})=\mathrm{w}_L(\pi(\mathbf{x}))=\sum_{i=1}^{n} \max\{\min (x_i,4-x_i), \min (x_{i+1}, 4-x_{i+1})\}=\sum_{i=1}^{n} \max\{\mathrm{w}_L(x_i), \mathrm{w}_L(x_{i+1})\}$, where $w_L(x_i)=\min\{x_i,4-x_i\}$ and $x_{n+1}=x_1$.
    \end{defn}
    \begin{defn}[\!\!\cite{PAA}]
    	Let $\mathbf{x}=(x_1, x_2, \ldots, x_n)$ and $\mathbf{y}=(y_1, y_2, \ldots, y_n) \in \mathbb{Z}_4^n$. The symbol-pair Lee-distance between $\mathbf{x}$ and $\mathbf{y}$ is defined as $d_p(\mathbf{x}, \mathbf{y})=\mathrm{w}_p(\mathbf{x}-\mathbf{y})=\mathrm{w}_L\big(\pi(\mathbf{x})-\pi(\mathbf{y})\big)=\sum_{i=1}^{n} \max\{\mathrm{w}_L\big((x_i-y_i) \mod 4\big), \mathrm{w}_L\big((x_{i+1}-y_{i+1})\mod 4\big)\}$, or equivalently, $d_p(\mathbf{x}, \mathbf{y})=\sum_{i=1}^{n} \max\{d_L(x_i, y_i), d_L(x_{i+1}, y_{i+1})\}$, where $x_{n+1}=x_1$ and $y_{n+1}=y_1$.
    \end{defn}
    The following example illustrates the computation of the symbol-pair Lee weight and symbol-pair Lee-distance.
    \begin{example}
    	Consider the vector $\mathbf{x}=(1,2,0,3)\in\mathbb Z_4^4$. The corresponding symbol-pair vector is
    	$\pi(\mathbf{x})=\big((1,2),(2,0),(0,3),(3,1)\big)$. Since $w_L(1)=1,\, w_L(2)=2,\, w_L(0)=0,\, w_L(3)=1,$
    	the symbol-pair Lee weight of $\mathbf{x}$ is $w_p(\mathbf{x})
    		=\max\{1,2\}
    		+\max\{2,0\}
    		+\max\{0,1\}
    		+\max\{1,1\}
    		=2+2+1+1
    		=6$. Now, let $\mathbf{y}=(0,1,1,3)\in\mathbb Z_4^4$. Then $\mathbf{x}-\mathbf{y}
    	=(1,1,3,0)\pmod 4$ and $\pi(\mathbf{x})-\pi(\mathbf{y})=\big((1,1),(1,3),(3,0),(0,1)\big)$.
    	The corresponding symbol-pair Lee-distance between $\mathbf{x}$ and $\mathbf{y}$ is
    	$d_p(\mathbf{x},\mathbf{y})
    		=w_p(\mathbf{x}-\mathbf{y})
    		=\max\{1,1\}
    		+\max\{1,1\}
    		+\max\{1,0\}
    		+\max\{0,1\}
    		=1+1+1+1
    		=4$.
    \end{example} The relationship between the Lee-distance and the symbol-pair Lee-distance over $\mathbb Z_4$ is given by the following lemma.
    \begin{lem}[\!\!{\cite{PAA}}]
    Let \( \boldsymbol{x}, \boldsymbol{y} \in \mathbb{Z}_4^n \). Then the following bounds hold:
    	\[
    	d_L(\mathbf{x},\mathbf{y})
    	\le
    	d_p(\mathbf{x},\mathbf{y})
    	\le
    	2d_L(\mathbf{x},\mathbf{y}).
    	\]
    \end{lem}
\section{Symbol-Pair Lee-Distance Codes over $\mathbb{Z}_q$} \label{section3}
The Lee-metric framework for symbol-pair read channels has previously been studied only over $\mathbb{Z}_4$. In this section, we generalize this notion to $\mathbb{Z}_q$, where $q \ge 2$. We define the symbol-pair Lee weight and symbol-pair Lee-distance over $\mathbb{Z}_q$, prove that the resulting distance function is a metric, and establish its relationship with the classical Lee-distance. This generalization provides a Lee-metric framework for $q$-ary symbol-pair read channels.

Let $\mathbf{x}=(x_1, x_2, \ldots, x_n)$ be a vector in $\mathbb{Z}_q^n$. Then the corresponding symbol-pair vector of $\mathbf{x}$ is defined as $\pi(\mathbf{x})=\big((x_1, x_2), (x_2, x_3), \ldots, (x_{n-1}, x_n), (x_n, x_1)\big)$, where $x_{n+1}=x_1$.
\begin{defn}
	 Let $\mathbf{x}=(x_1, x_2, \ldots, x_n) \in \mathbb{Z}_q^n$. The symbol-pair Lee weight of $\mathbf{x}$ is defined as $\mathrm{w}_p(\mathbf{x})=\mathrm{w}_L(\pi(\mathbf{x}))=\sum_{i=1}^{n} \max\{\min (x_i,q-x_i), \min (x_{i+1}, q-x_{i+1})\}=\sum_{i=1}^{n} \max\{\mathrm{w}_L(x_i), \mathrm{w}_L(x_{i+1})\}$, where $w_L(x_i)=\min\{x_i,q-x_i\}$ and $x_{n+1}=x_1$.
\end{defn}
\begin{defn}
	 Let $\mathbf{x}=(x_1, x_2, \ldots, x_n)$ and $\mathbf{y}=(y_1, y_2, \ldots, y_n) \in \mathbb{Z}_q^n$. The symbol-pair Lee-distance between $\mathbf{x}$ and $\mathbf{y}$ is defined as $d_p(\mathbf{x}, \mathbf{y})=\mathrm{w}_p(\mathbf{x}-\mathbf{y})=\mathrm{w}_L\big(\pi(\mathbf{x})-\pi(\mathbf{y})\big)=\sum_{i=1}^{n} \max\{\mathrm{w}_L\big((x_i-y_i) \mod q\big), \mathrm{w}_L\big((x_{i+1}-y_{i+1})\mod q\big)\}$, or equivalently, $d_p(\mathbf{x}, \mathbf{y})=\sum_{i=1}^{n} \max\{d_L(x_i, y_i), d_L(x_{i+1}, y_{i+1})\}$, where $x_{n+1}=x_1$ and $y_{n+1}=y_1$.
\end{defn}
The symbol-pair Lee-distance was previously shown to be a metric only over $\mathbb{Z}_4$ in \cite{PAA}. Since the present work extends the symbol-pair Lee-distance to arbitrary alphabets $\mathbb{Z}_q$, where $q \ge 2$, it is necessary to verify that the generalized distance function also satisfies the metric axioms. The following theorem establishes this result. 
\begin{thm}
	The symbol-pair Lee-distance $d_p:\mathbb Z_q^n\times\mathbb Z_q^n\to\mathbb N_0$ is a metric on $\mathbb Z_q^n$.
\end{thm}
\begin{IEEEproof}
Let $\mathbf{x},\mathbf{y},\mathbf{z}\in\mathbb Z_q^n$. We verify the four metric axioms.

(i) \textit{Non-negativity}:
Since the Lee-distance is non-negative, $d_L(x_i,y_i)\ge 0$ for every $i\in[n]$. Therefore,
$\max\{d_L(x_i,y_i),d_L(x_{i+1},y_{i+1})\}\ge 0$ for every $i\in[n]$. Hence,
\[
d_p(\mathbf{x},\mathbf{y})
=
\sum_{i=1}^{n}
\max\{d_L(x_i,y_i),d_L(x_{i+1},y_{i+1})\}
\ge 0.
\]

(ii) \textit{Identity of indiscernibles}:
Suppose that $d_p(\mathbf{x},\mathbf{y})=0$. Since each term in the sum defining $d_p(\mathbf{x},\mathbf{y})$ is non-negative, it follows that
$\max\{d_L(x_i,y_i),d_L(x_{i+1},y_{i+1})\}=0$ for all $i\in[n]$. Since the maximum of two non-negative numbers is zero only when both numbers are zero, we have $d_L(x_i,y_i)=0$ for every $i\in[n]$. As $d_L$ is a metric on $\mathbb Z_q$, it follows that $x_i=y_i$ for all $i\in[n]$, and hence $\mathbf{x}=\mathbf{y}$. Conversely, if $\mathbf{x}=\mathbf{y}$, then $d_L(x_i,y_i)=0$ for all $i\in[n]$, and therefore $d_p(\mathbf{x},\mathbf{y})=0$. Thus,
\[
d_p(\mathbf{x},\mathbf{y})=0
\iff
\mathbf{x}=\mathbf{y}.
\]

(iii) \textit{Symmetry}:
Since the Lee-distance is symmetric, $d_L(x_i,y_i)=d_L(y_i,x_i)$ for all $i\in[n]$. Hence,
\[
\begin{aligned}
	d_p(\mathbf{x},\mathbf{y})
	&=
	\sum_{i=1}^{n}
	\max\{d_L(x_i,y_i),d_L(x_{i+1},y_{i+1})\} \\
	&=
	\sum_{i=1}^{n}
	\max\{d_L(y_i,x_i),d_L(y_{i+1},x_{i+1})\} \\
	&=
	d_p(\mathbf{y},\mathbf{x}).
\end{aligned}
\]

(iv) \textit{Triangle inequality}:
For each $i\in[n]$, the Lee-distance satisfies $d_L(x_i,z_i)
\le
d_L(x_i,y_i)+d_L(y_i,z_i)$
and $d_L(x_{i+1},z_{i+1})
\le
d_L(x_{i+1},y_{i+1})
+
d_L(y_{i+1},z_{i+1})$. Using the elementary inequality $\max\{a+b,c+d\}
\le
\max\{a,c\}
+
\max\{b,d\}$,
which follows since $a+b \le \max\{a,c\}+\max\{b,d\}$ and $c+d \le \max\{a,c\}+\max\{b,d\}$. Hence, we have
\[
\begin{aligned}
	\max\{d_L(x_i,z_i),d_L(x_{i+1},z_{i+1})\}
	&\le
	\max\{
	d_L(x_i,y_i)+d_L(y_i,z_i),
	d_L(x_{i+1},y_{i+1})+d_L(y_{i+1},z_{i+1})
	\} \\
	&\le
	\max\{d_L(x_i,y_i),d_L(x_{i+1},y_{i+1})\}
	+
	\max\{d_L(y_i,z_i),d_L(y_{i+1},z_{i+1})\}.
\end{aligned}
\]
Summing over all $i\in[n]$ yields
\[
\begin{aligned}
	d_p(\mathbf{x},\mathbf{z})
	&=
	\sum_{i=1}^{n}
	\max\{d_L(x_i,z_i),d_L(x_{i+1},z_{i+1})\} \\
	&\le
	\sum_{i=1}^{n}
	\max\{d_L(x_i,y_i),d_L(x_{i+1},y_{i+1})\} \\
	&\quad+
	\sum_{i=1}^{n}
	\max\{d_L(y_i,z_i),d_L(y_{i+1},z_{i+1})\} \\
	&=
	d_p(\mathbf{x},\mathbf{y})
	+
	d_p(\mathbf{y},\mathbf{z}).
\end{aligned}
\]
Therefore, the triangle inequality holds. Since $d_p$ satisfies \textit{non-negativity}, \textit{identity of indiscernibles}, \textit{symmetry}, and the \textit{triangle inequality}, it follows that $d_p$ is a metric on $\mathbb Z_q^n$.
\end{IEEEproof}
While a relationship between the classical Lee-distance and the symbol-pair Lee-distance was established in \cite{PAA} for codes over $\mathbb{Z}_4$, the proof relies on a case-by-case analysis of symbol occurrences specific to the four-element alphabet. The following lemma extends this relationship to $\mathbb{Z}_q$, where $q \ge 2$, based on the generalized definition of the symbol-pair Lee-distance and admits a simple proof that is independent of the alphabet size.  
\begin{lem} \label{dpdl}
	Let $\mathbf{x},\mathbf{y}\in\mathbb Z_q^n$. Then the following bounds hold:
	\[
	d_L(\mathbf{x},\mathbf{y})
	\le
	d_p(\mathbf{x},\mathbf{y})
	\le
	2d_L(\mathbf{x},\mathbf{y}).
	\]
\end{lem}

\begin{IEEEproof}
	Let $a_i=d_L(x_i,y_i), \, i\in[n]$.
	Then $d_L(\mathbf{x},\mathbf{y})
	=
	\sum_{i=1}^{n} a_i$,
	and $d_p(\mathbf{x},\mathbf{y})
	=
	\sum_{i=1}^{n}
	\max\{a_i,a_{i+1}\}$,
	where $a_{n+1}=a_1$. To prove the lower bound, note that $\max\{a_i,a_{i+1}\}\ge a_i$
	for every $i\in[n]$. Summing over all $i$ gives
	\[
	d_p(\mathbf{x},\mathbf{y})
	=
	\sum_{i=1}^{n}
	\max\{a_i,a_{i+1}\}
	\ge
	\sum_{i=1}^{n} a_i
	=
	d_L(\mathbf{x},\mathbf{y}).
	\]
	
	For the upper bound, observe that $\max\{a_i,a_{i+1}\}
	\le
	a_i+a_{i+1}$
	for every $i\in[n]$. Therefore,
	\[
	\begin{aligned}
		d_p(\mathbf{x},\mathbf{y})
		&=
		\sum_{i=1}^{n}
		\max\{a_i,a_{i+1}\} \\
		&\le
		\sum_{i=1}^{n}
		(a_i+a_{i+1}) \\
		&=
		\sum_{i=1}^{n} a_i
		+
		\sum_{i=1}^{n} a_{i+1}.
	\end{aligned}
	\]
	Since $a_{n+1}=a_1$, we have
	\[
	\sum_{i=1}^{n} a_{i+1}
	=
	\sum_{i=1}^{n} a_i,
	\]
	and hence
	\[
	d_p(\mathbf{x},\mathbf{y})
	\le
	2\sum_{i=1}^{n} a_i
	=
	2d_L(\mathbf{x},\mathbf{y}).
	\]
	
	Combining the two inequalities yields
	\[
	d_L(\mathbf{x},\mathbf{y})
	\le
	d_p(\mathbf{x},\mathbf{y})
	\le
	2d_L(\mathbf{x},\mathbf{y}).
	\]
\end{IEEEproof}
	The lower bound in Lemma~\ref{dpdl} is attained if and only if $d_L(x_1,y_1)=d_L(x_2,y_2)=\cdots=d_L(x_n,y_n)$.
	Let $a_i=d_L(x_i,y_i)$ for $i\in[n]$. Then $d_p(\mathbf{x},\mathbf{y})
	=
	d_L(\mathbf{x},\mathbf{y})$
	if and only if $\max\{a_i,a_{i+1}\}=a_i,
	\, i\in[n],$
	where $a_{n+1}=a_1$. Equivalently, $\, a_i\ge a_{i+1},
	\, i\in[n]$.
	Since $a_{n+1}=a_1$, we obtain $a_1\ge a_2\ge \cdots \ge a_n\ge a_1,$
	which holds if and only if $a_1=a_2=\cdots=a_n$. 
	Therefore, $d_L(x_1,y_1)=d_L(x_2,y_2)=\cdots=d_L(x_n,y_n)$.
 The following definition formalizes the notion of \(t\)-pair errors under the Lee metric.
\begin{defn}
	Let $\mathbf{x}=(x_1,x_2,\ldots,x_n)\in\mathbb Z_q^n$. A symbol-pair vector $$\overset{\leftrightarrow}{\mathbf{y}}=\big((\overset{\leftarrow}{y_1}, \overset{\rightarrow}{y_1}),\ldots, (\overset{\leftarrow}{y_n}, \overset{\rightarrow}{y_n})\big)\in (\mathbb Z_q\times \mathbb Z_q)^n$$ is the result of at most \(t\) symbol-pair errors from
    $\mathbf{x}$ if $d_p\bigl(\pi(\mathbf{x}),\overset{\leftrightarrow}{\mathbf{y}}\bigr)\le t$,
where \(d_p\) denotes the symbol-pair Lee-distance, and
\(\overset{\leftarrow}{y_i}\) and \(\overset{\rightarrow}{y_i}\) denote the left and right symbols, respectively, of the \(i\)th symbol pair in \(\overset{\leftrightarrow}{\mathbf{y}}\), for \(i\in[n]\).
\end{defn} Throughout this paper, $t$ denotes the maximum number of symbol-pair Lee errors or simply pair errors.
The following lemma characterizes the error-correcting capability of Lee-distance codes for symbol-pair read channels in terms of their minimum symbol-pair Lee-distance.
\begin{lem}
	A Lee-distance code for symbol-pair read channels with minimum symbol-pair Lee-distance \( d_p \) can correct up to \( t \) symbol-pair errors if and only if
	\(
	d_p \geq 2t + 1.
	\)
\end{lem}
\section{General Redundancy Bounds for FCSPLCs} \label{section4}
The generalization of the symbol-pair Lee-distance from $\mathbb{Z}_4$ to $\mathbb{Z}_q$, where $q \ge 2$, provides a Lee-metric framework for $q$-ary symbol-pair read channels. Using this framework for $q=2^m$, where $m\ge 1$, we introduce function-correcting symbol-pair Lee-distance codes (FCSPLCs) and investigate their redundancy requirements through their connection to irregular-pair Lee-distance codes. We derive general lower and upper bounds on the optimal redundancy of FCSPLCs in terms of the shortest length of irregular-pair Lee-distance codes, and subsequently obtain simplified bounds that are easier to evaluate. Finally, by establishing lower and upper bounds on the achievable length of irregular-pair Lee-distance codes, we derive Plotkin-type and Gilbert--Varshamov-type bounds on the optimal redundancy of FCSPLCs. For bijective functions, these results yield corresponding Plotkin-type and Gilbert--Varshamov-type bounds for classical Lee-distance codes over symbol-pair read channels.
    \subsection{FCSPLCs and Irregular-Pair Lee-distance Codes}
    Let $\mathbf{u}\in\mathbb Z_q^k$ be a message and let
    $f:\mathbb Z_q^k\rightarrow \operatorname{Im}(f) \triangleq \{ f(\mathbf{u}): \mathbf{u} \in \mathbb Z_q^k\}$ be a function computed on $\mathbf{u}$, where
    $E=|\operatorname{Im}(f)|\le q^k$ is the expressiveness of $f$.
    The message is encoded via an encoding function
    \[
    \mathrm{Enc}:\mathbb Z_q^k\rightarrow \mathbb Z_q^{k+r},
    \qquad
    \mathrm{Enc}(\mathbf{u})
    =
    \big(\mathbf{u},\mathbf{p}(\mathbf{u})\big),
    \]
    where $\mathbf{p}(\mathbf{u})\in\mathbb Z_q^r$ is the redundancy
    vector. The resulting codeword $\mathrm{Enc}(\mathbf{u})$ is transmitted over a symbol-pair
    read channel, resulting in a received vector 
    $\overset{\leftrightarrow}{\mathbf{y}}\in(\mathbb Z_q\times\mathbb Z_q)^{k+r}$ satisfying
    $d_p(\pi(\mathrm{Enc}(\mathbf{u})),\overset{\leftrightarrow}{\mathbf{y}})\le t$.
    The formal definition of an FCSPLC is given below.
    \begin{defn} \label{FCSPLC}
    	An encoding function
    	\(\mathrm{Enc}:\mathbb Z_q^k\to \mathbb Z_q^{k+r}\), $\mathrm{Enc}(\mathbf{u})
    	=
    	\big(\mathbf{u},\mathbf{p}(\mathbf{u})\big)$ defines a function-correcting symbol-pair Lee-distance code (FCSPLC) for the function \(f:\mathbb Z_q^k\to \operatorname{Im}(f)\) if for all
    	\(\mathbf u,\mathbf v\in\mathbb Z_q^k\) satisfying
    	\(f(\mathbf u)\neq f(\mathbf v)\), we have 
    	\[
    	d_p\big(\mathrm{Enc}(\mathbf u),
    	\mathrm{Enc}(\mathbf v)\big)
    	\ge 2t+1.
    	\]
    \end{defn}
    By definition~\ref{FCSPLC}, for any received symbol-pair vector $\overset{\leftrightarrow}{\mathbf{y}}$ obtained from $\mathrm{Enc}(\mathbf{u})$ through at most $t$-pair errors, the receiver can uniquely recover $f(\mathbf{u})$, provided that both the function $f$ and the encoding function  $\mathrm{Enc}$ are known.
    Given an information vector $\mathbf{u} \in \mathbb{Z}_q^k$, the encoder computes the parity vector $\mathbf{p}(\mathbf{u})\in\mathbb Z_q^r$ according to the prescribed construction and outputs the systematic codeword  $\mathrm{Enc}(\mathbf{u})=\big(\mathbf{u},\mathbf{p}(\mathbf{u})\big) \in \mathbb Z_q^{k+r}$. The parity vector is constructed such that for any pair $(\mathbf{u}, \mathbf{v})$ with \(f(\mathbf u)\neq f(\mathbf v)\), the corresponding codewords satisfy $d_p\big(\mathrm{Enc}(\mathbf u),
    \mathrm{Enc}(\mathbf v)\big)\ge 2t+1$.
    
    Let $\overset{\leftrightarrow}{\mathbf{y}}$ be the received symbol-pair vector over a $q$-ary symmetric channel under the Lee metric. The decoder performs minimum symbol-pair Lee-distance decoding, i.e., $\hat{\mathbf{u}} = \displaystyle \arg\min_{\mathbf{u} \in \mathbb{Z}_q^k}
    d_p\big(\overset{\leftrightarrow}{\mathbf{y}}, \mathrm{Enc}(\mathbf{u})\big)$. The decoder then outputs the estimated function value $f(\hat{\mathbf{u}})$. Since the code satisfies $d_p\big(\mathrm{Enc}(\mathbf{u}), \mathrm{Enc}(\mathbf{v})\big) \geq 2t + 1$ whenever  $f(\mathbf{u}) \neq f(\mathbf{v})$, it follows that for any pair error vector of Lee weight at most $t$, the minimum symbol-pair Lee-distance decoder uniquely determines the correct function value. Next, we define the term \emph{optimal redundancy} of an FCSPLC for the function $f$.
    \begin{defn}
    	The optimal redundancy \( r_{p}^f(q, k, t) \) is defined as the smallest integer \( r \) for which there exists an FCSPLC with an encoding function \( \mathrm{Enc} : \mathbb{Z}_q^k \to \mathbb{Z}_q^{k+r} \) that enables recovery of \( f(\mathbf{u}) \) from a received symbol-pair vector $\overset{\leftrightarrow}{\mathbf{y}}$ satisfying $d_p\big(\overset{\leftrightarrow}{\mathbf{y}}, \mathrm{Enc}(\mathbf{u})\big) \le t$.
    \end{defn}
    We establish a connection between FCSPLCs and irregular-pair Lee-distance codes, which forms the foundation for deriving bounds on the optimal redundancy. Towards this, we introduce two pair Lee-distance matrices associated with a function $f$.
    \begin{defn}
    	Let $\mathbf{u}_1, \dots, \mathbf{u}_M \in \mathbb{Z}_{q}^{k}$. We define two pair Lee-distance matrices (DRMs) $\mathbf{D}_f^{(1)}(t, \mathbf{u}_1, \dots, \mathbf{u}_M)$ and $\mathbf{D}_f^{(2)}(t, \mathbf{u}_1, \dots, \mathbf{u}_M)$, each of size $M \times M$, whose entries are given by 
    	\begin{align*} \left [{{\boldsymbol {D}_{f}^{(1)}\left ({{t, \boldsymbol {u}_{1}, \ldots, \boldsymbol {u}_{M}}}\right)}}\right ]_{i j}= \begin{cases}{\left [{{\big(2t+1- \frac q2 \big)-d_{p}\left ({{\boldsymbol {u}_{i}, \boldsymbol {u}_{j}}}\right)}}\right ]^{+},} & \text {if}~ f\left ({{\boldsymbol {u}_{i}}}\right) \neq f\left ({{\boldsymbol {u}_{j}}}\right), \\ 0, & \text { otherwise.}\end{cases}\end{align*} and 
    	\begin{align*} \ \left [{{\boldsymbol {D}_{f}^{(2)}\left ({{t, \boldsymbol {u}_{1}, \ldots, \boldsymbol {u}_{M}}}\right)}}\right ]_{i j}\!\!=\!\! \begin{cases}{\left [{{\big(2t+1+ \frac q2\big)-d_{p}\left ({{\boldsymbol {u}_{i}, \boldsymbol {u}_{j}}}\right)}}\right ]^{+},} & \text {if}~ f\left ({{\boldsymbol {u}_{i}}}\right) \neq f\left ({{\boldsymbol {u}_{j}}}\right), \\ 0, & \text { otherwise.}~\end{cases}\end{align*}
    \end{defn}
    Let $\mathcal{P}=\{\mathbf{p}_1, \dots, \mathbf{p}_M\} \subseteq \mathbb{Z}_q^r$ be a code of length $r$ and cardinality $M$. Given a matrix $\mathbf{D} \in \mathbb{N}_0^{M\times M}$, irregular-pair Lee-distance codes are defined as follows.
    \begin{defn} \label{D_p-code}
    	Let $\mathbf{D} \in \mathbb{N}_0^{M\times M}$. Then $\mathcal{P}=\{\mathbf{p}_1, \dots, \mathbf{p}_M\} \subseteq \mathbb{Z}_q^r$ is a $\mathbf{D}$-irregular-pair Lee-distance code $(\mathbf{D}_p$-code$)$, if there exists an ordering of the codewords of  $\mathcal{P}$ such that $d_p(\mathbf{p}_i, \mathbf{p}_j) \ge [\mathbf{D}]_{ij}$ for all $i, j \in [M]$.
    \end{defn}
    The smallest integer $r$ such that there exists a $\mathbf{D}_{p}$-code of length $r$ is denoted by $N_{p}(\mathbf{D})$. If $[\mathbf{D}]_{ij} = d$ for all $i \neq j$, we simply write $N_{p}(M,d)$.
    By definition, a $\mathbf{D}_{p}$-code imposes individual pair Lee-distance constraints between every pair of codewords. The following example illustrates the definition of DRMs and $\mathbf{D}_p$-code.
    \begin{example}
    	Let $\mathbf{u}_i = (\mathbf{u}_{i1}, \mathbf{u}_{i2}) \in \mathbb{Z}_8^2, \;\forall i \in [5]$. Consider the symbol-pair Lee weight function $f:\mathbb Z_8^2 \rightarrow \mathbb N_0$,
    	defined by $f(\mathbf u_i)=w_p(\mathbf u_i)$ for $\mathbf u_1=(0,0),\,
    	\mathbf u_2=(1,0),\,
    	\mathbf u_3=(0,1),\,
    	\mathbf u_4=(2,0),\,
    	\mathbf u_5=(1,1)$. Since $w_L(0)=0,\, w_L(1)=1,\, w_L(2)=2$,
    	the corresponding symbol-pair Lee weights are $f(\mathbf u_1)=w_p(\mathbf u_1)=0$,
    	$f(\mathbf u_2)=w_p(\mathbf u_2)
    	=
    	2\max\{1,0\}
    	=2$, $f(\mathbf u_3)=w_p(\mathbf u_3)
    	=
    	2\max\{0,1\}
    	=2$, $f(\mathbf u_4)=w_p(\mathbf u_4)
    	=
    	2\max\{2,0\}
    	=4$, 
    	and $f(\mathbf u_5)=w_p(\mathbf u_5)
    	=
    	2\max\{1,1\}
    	=2$. Thus, $f(\mathbf u_2)=f(\mathbf u_3)=f(\mathbf u_5)=2$,
    	while the remaining function values are distinct.
    	
    	The pair Lee-distances between the vectors are given by
    	\[
    	[d_p(\mathbf u_i,\mathbf u_j)]_{i,j}
    	=
    	\begin{blockarray}{c@{\hskip 4pt}*{5}{c}}
    		& \mathbf{u}_1 & \mathbf{u}_2 & \mathbf{u}_3 & \mathbf{u}_4 & \mathbf{u}_5 \\
    		\begin{block}{c@{\hskip 4pt}[*{5}{c}]}
    			\mathbf{u}_1 \text{ } & 0 & 2 & 2 & 4 & 2 \\
    			\mathbf{u}_2 \text{ } & 2 & 0 & 2 & 2 & 2 \\
    			\mathbf{u}_3 \text{ } & 2 & 2 & 0 & 4 & 2 \\
    			\mathbf{u}_4 \text{ } & 4 & 2 & 4 & 0 & 4 \\
    			\mathbf{u}_5 \text{ } & 2 & 2 & 2 & 4 & 0 \\
    		\end{block}
    	\end{blockarray}.
    	\]
    	For $t=3$,
    	the corresponding DRMs are
    	\[
    	\mathbf D_f^{(1)}(t=3,\mathbf{u}_1,\ldots,\mathbf{u}_5)
    	=
    	\begin{blockarray}{c@{\hskip 4pt}*{5}{c}}
    		& \mathbf{u}_1 & \mathbf{u}_2 & \mathbf{u}_3 & \mathbf{u}_4 & \mathbf{u}_5 \\
    		\begin{block}{c@{\hskip 4pt}[*{5}{c}]}
    			\mathbf{u}_1 \text{ } & 0 & 1 & 1 & 0 & 1 \\
    			\mathbf{u}_2 \text{ } & 1 & 0 & 0 & 1 & 0 \\
    			\mathbf{u}_3 \text{ } & 1 & 0 & 0 & 0 & 0 \\
    			\mathbf{u}_4 \text{ } & 0 & 1 & 0 & 0 & 1 \\
    			\mathbf{u}_5 \text{ } & 1 & 0 & 0 & 1 & 0 \\
    		\end{block}
    	\end{blockarray},
    	\] and 
    	\[
    	\mathbf D_f^{(2)}(t=3,\mathbf{u}_1,\ldots,\mathbf{u}_5)
    	=
    	\begin{blockarray}{c@{\hskip 4pt}*{5}{c}}
    		& \mathbf{u}_1 & \mathbf{u}_2 & \mathbf{u}_3 & \mathbf{u}_4 & \mathbf{u}_5 \\
    		\begin{block}{c@{\hskip 4pt}[*{5}{c}]}
    			\mathbf{u}_1 \text{ } & 0 & 9 & 9 & 7 & 9 \\
    			\mathbf{u}_2 \text{ } & 9 & 0 & 0 & 9 & 0 \\
    			\mathbf{u}_3 \text{ } & 9 & 0 & 0 & 0 & 0 \\
    			\mathbf{u}_4 \text{ } & 7 & 9 & 0 & 0 & 7 \\
    			\mathbf{u}_5 \text{ } & 9 & 0 & 0 & 7 & 0 \\
    		\end{block}
    	\end{blockarray}.
    	\] Observe that the entries corresponding to the pairs
    	$(\mathbf u_2,\mathbf u_3)$, $(\mathbf u_2,\mathbf u_5)$, and
    	$(\mathbf u_3,\mathbf u_5)$ are zero, since these vectors have the same
    	function value.
    	
    	Now consider the code $\mathcal P
    	=
    	\{\mathbf p_1,\mathbf p_2,\mathbf p_3,\mathbf p_4,\mathbf p_5\}
    	=
    	\{0,1,2,3,4\}
    	\subseteq \mathbb Z_8$. Since the symbol-pair Lee-distance of length-$1$ vectors coincides with the
    	Lee-distance, we have
    	\[
    	[d_p(\mathbf p_i,\mathbf p_j)]_{i,j}
    	=
    	\begin{blockarray}{c@{\hskip 4pt}*{5}{c}}
    		& \mathbf{p}_1 & \mathbf{p}_2 & \mathbf{p}_3 & \mathbf{p}_4 & \mathbf{p}_5 \\
    		\begin{block}{c@{\hskip 4pt}[*{5}{c}]}
    			\mathbf{p}_1 \text{ } & 0 & 1 & 2 & 3 & 4 \\
    			\mathbf{p}_2 \text{ } & 1 & 0 & 1 & 2 & 3 \\
    			\mathbf{p}_3 \text{ } & 2 & 1 & 0 & 1 & 2 \\
    			\mathbf{p}_4 \text{ } & 3 & 2 & 1 & 0 & 1 \\
    			\mathbf{p}_5 \text{ } & 4 & 3 & 2 & 1 & 0 \\
    		\end{block}
    	\end{blockarray}.
    	\]Comparing this matrix with $\mathbf D_f^{(1)}$, we obtain $d_p(\mathbf p_i,\mathbf p_j)
    	\ge
    	[\mathbf D_f^{(1)}]_{ij},
    	\,
    	\forall\, i,j\in[5]$. Therefore, $\mathcal P$ is a
    	$\mathbf D_p^{(1)}$-irregular-pair Lee-distance code.
    \end{example}
    With these definitions, we now proceed to establish a link between FCSLCs and irregular-pair Lee-distance codes. The following lemma plays an important role in the proof of Theorem~\ref{Red_Bounds}.
    \begin{lem} \label{lem_delta}
    	Let $\mathbf{u}
    	=
    	\bigl(\mathbf{u}^{(1)},\mathbf{u}^{(2)}\bigr)
    	\in \mathbb{Z}_q^{m+r}
    	\, \text{and} \,
    	\mathbf{v}
    	=
    	\bigl(\mathbf{v}^{(1)},\mathbf{v}^{(2)}\bigr)
    	\in \mathbb{Z}_q^{m+r}$, 
    	where $\mathbf{u}^{(1)}
    	=
    	(u_1,\ldots,u_m),
    	\,
    	\mathbf{u}^{(2)}
    	=
    	(u_{m+1},\ldots,u_{m+r})$, 
    	and $\mathbf{v}^{(1)}
    	=
    	(v_1,\ldots,v_m),
    	\,
    	\mathbf{v}^{(2)}
    	=
    	(v_{m+1},\ldots,v_{m+r})$.
    	Then
    	\[
    	d_p\bigl(\mathbf{u}^{(1)},\mathbf{v}^{(1)}\bigr)
    	+
    	d_p\bigl(\mathbf{u}^{(2)},\mathbf{v}^{(2)}\bigr)
    	-\frac{q}{2} \,
    	\le \,
    	d_p(\mathbf{u},\mathbf{v}) \,
    	\le \,
    	d_p\bigl(\mathbf{u}^{(1)},\mathbf{v}^{(1)}\bigr)
    	+
    	d_p\bigl(\mathbf{u}^{(2)},\mathbf{v}^{(2)}\bigr)
    	+\frac{q}{2}.
    	\]
    \end{lem}
    \begin{IEEEproof}
    		Let $a_i=d_L(u_i,v_i), \, i=1,\ldots,m+r$.
    		By the definition of the symbol-pair Lee-distance,
    		\[
    		\begin{aligned}
    			d_p(\mathbf{u},\mathbf{v})
    			&=
    			\sum_{i=1}^{m-1}\max\{a_i,a_{i+1}\}
    			+\max\{a_m,a_{m+1}\} 
    			&\,+
    			\sum_{i=m+1}^{m+r-1}\max\{a_i,a_{i+1}\}
    			+\max\{a_{m+r},a_1\},
    		\end{aligned}
    		\]
    		while
    		\[
    		d_p\bigl(\mathbf{u}^{(1)},\mathbf{v}^{(1)}\bigr)
    		=
    		\sum_{i=1}^{m-1}\max\{a_i,a_{i+1}\}
    		+\max\{a_m,a_1\},
    		\]
    		and
    		\[
    		d_p\bigl(\mathbf{u}^{(2)},\mathbf{v}^{(2)}\bigr)
    		=
    		\sum_{i=m+1}^{m+r-1}\max\{a_i,a_{i+1}\}
    		+\max\{a_{m+r},a_{m+1}\}.
    		\] Therefore,
    		\begin{equation} \label{delta}
    			\begin{aligned}
    				\Delta
    				&\triangleq
    				d_p(\mathbf{u},\mathbf{v})
    				-d_p\bigl(\mathbf{u}^{(1)},\mathbf{v}^{(1)}\bigr)
    				-d_p\bigl(\mathbf{u}^{(2)},\mathbf{v}^{(2)}\bigr)\\
    				&=
    				\max\{a_m,a_{m+1}\}
    				+\max\{a_{m+r},a_1\}
    				-\max\{a_m,a_1\}
    				-\max\{a_{m+r},a_{m+1}\}.
    			\end{aligned}
    		\end{equation}
    		
    		Let $a=a_m,\,
    		b=a_{m+1},\,
    		c=a_{m+r},\,
    		d=a_1$.
    		Then $\Delta
    		=
    		\max\{a,b\}
    		+\max\{c,d\}
    		-\max\{a,d\}
    		-\max\{c,b\}$. Using the identity $\max\{x,y\}
    		=
    		\frac{x+y+|x-y|}{2}$,
    		we obtain
    		\[
    		\begin{aligned}
    			\Delta
    			&=
    			\frac{
    				(a+b+|a-b|)
    				+(c+d+|c-d|)
    			}{2}-
    			\frac{
    				(a+d+|a-d|)
    				+(c+b+|c-b|)
    			}{2}\\
    			&=
    			\frac{
    				|a-b|
    				+|c-d|
    				-|a-d|
    				-|c-b|
    			}{2}.
    		\end{aligned}
    		\] By the triangle inequality, $|a-b|
    		\le
    		|a-d|+|d-b|$,
    		and $|c-d|
    		\le
    		|c-b|+|b-d|$.
    		Adding the above inequalities yields $|a-b|+|c-d|
    		\le
    		|a-d|+|c-b|+2|b-d|$.
    		Hence, $\Delta
    		\le
    		|b-d|$. Similarly, $|a-d|
    		\le
    		|a-b|+|b-d|$,
    		and $|c-b|
    		\le
    		|c-d|+|d-b|$.
    		Adding these inequalities gives $|a-d|+|c-b|
    		\le
    		|a-b|+|c-d|+2|b-d|$, 
    		which implies $\Delta
    		\ge
    		-|b-d|$. Therefore, $-|b-d|
    		\le
    		\Delta
    		\le
    		|b-d|$. Since \(a_i=d_L(u_i,v_i)\) for all \(i\), we have $0\le a_i\le \frac{q}{2}$.
    		Consequently, $|b-d|
    		\le
    		\frac{q}{2}$,
    		and hence $-\frac{q}{2}
    		\le
    		\Delta
    		\le
    		\frac{q}{2}$. Thus (\ref{delta}) yields \[
    		d_p\bigl(\mathbf{u}^{(1)},\mathbf{v}^{(1)}\bigr)
    		+
    		d_p\bigl(\mathbf{u}^{(2)},\mathbf{v}^{(2)}\bigr)
    		-\frac{q}{2} \,
    		\le \,
    		d_p(\mathbf{u},\mathbf{v}) \,
    		\le \,
    		d_p\bigl(\mathbf{u}^{(1)},\mathbf{v}^{(1)}\bigr)
    		+
    		d_p\bigl(\mathbf{u}^{(2)},\mathbf{v}^{(2)}\bigr)
    		+\frac{q}{2}.
    		\]
    \end{IEEEproof} 
   For $q=2$, the Lee-distance coincides with the Hamming distance, and hence the symbol-pair Lee-distance reduces to the symbol-pair Hamming distance. Consequently, Lemma~\ref{lem_delta} specializes to
   \[
   d_p\bigl(\mathbf{u}^{(1)},\mathbf{v}^{(1)}\bigr)
   +
   d_p\bigl(\mathbf{u}^{(2)},\mathbf{v}^{(2)}\bigr)
   -1
   \le
   d_p(\mathbf{u},\mathbf{v})
   \le
   d_p\bigl(\mathbf{u}^{(1)},\mathbf{v}^{(1)}\bigr)
   +
   d_p\bigl(\mathbf{u}^{(2)},\mathbf{v}^{(2)}\bigr)
   +1,
   \]
   which is precisely the corresponding result for binary symbol-pair codes in \cite{XLC}. The following corollary considers the case where the parity vectors are the $r$-fold repetition vectors $(0,\ldots,0)$ and $\left(\frac{q}{2},\ldots,\frac{q}{2}\right)$, respectively.
    \begin{cor}\label{cor_delta2}
    	Let $\operatorname{Enc}(\mathbf{u})
    	=
    	(\mathbf{u},\{0\}^{r})
    	\quad\text{and}\quad
    	\operatorname{Enc}(\mathbf{v})
    	=
    	\left(\mathbf{v},\left\{\frac q2\right\}^{r}\right)$ for $\mathbf u, \mathbf v \in \mathbb{Z}_q^k.$
    	Then
    	\[
    	d_p(\operatorname{Enc}(\mathbf{u}),\operatorname{Enc}(\mathbf{v}))
    	\ge
    	d_p(\mathbf{u},\mathbf{v})
    	+
    	d_p\!\left(\{0\}^{r},
    	\left\{\frac q2\right\}^{r}\right).
    	\]
    \end{cor}
    
    \begin{IEEEproof}
    	Applying Lemma~\ref{lem_delta}, we have $d_p(\operatorname{Enc}(\mathbf{u}),\operatorname{Enc}(\mathbf{v}))
    	=
    	d_p(\mathbf{u},\mathbf{v})
    	+
    	d_p\!\left(\{0\}^{r},
    	\left\{\frac q2\right\}^{r}\right)
    	+\Delta,$
    	where $\Delta
    	=
    	\max\{a,b\}
    	+
    	\max\{b,d\}
    	-
    	\max\{a,d\}
    	-
    	b,$
    	with $a=d_L(u_k,v_k),\,
    	b=d_L\!\left(0,\frac q2\right)=\frac q2,\,
    	d=d_L(u_1,v_1).$
    	Since $a,d\le\frac q2=b,$
    	it follows that $\max\{a,b\}
    	=
    	\max\{b,d\}
    	=
    	b.$
    	Hence, $\Delta
    	=
    	b+b-\max\{a,d\}-b
    	=
    	b-\max\{a,d\}
    	\ge0.$
    	Therefore, $d_p(\operatorname{Enc}(\mathbf{u}),\operatorname{Enc}(\mathbf{v}))
    	\ge
    	d_p(\mathbf{u},\mathbf{v})
    	+
    	d_p\!\left(\{0\}^{r},
    	\left\{\frac q2\right\}^{r}\right).$
    \end{IEEEproof}
    \begin{thm} \label{Red_Bounds}
    	For any function $f: \mathbb{Z}_q^k \to Im(f)$ and $\{\mathbf{u}_1, \dots, \mathbf{u}_{q^k}\}=\mathbb{Z}_q^k$, we have \begin{equation*} N_{p}\big (\boldsymbol {D}_{f}^{(1)}(t,\boldsymbol {u}_{1},\ldots,\boldsymbol {u}_{q^{k}})\big) \,\leq \, r_{p}^{f}(k,t) \, \leq \, N_{p}\big (\boldsymbol {D}_{f}^{(2)}(t,\boldsymbol {u}_{1},\ldots,\boldsymbol {u}_{q^{k}})\big).\end{equation*}
    \end{thm} 
    \begin{IEEEproof}
    		Let $\mathbf{D}^{(1)}
    		\triangleq
    		\mathbf{D}_f^{(1)}(t,\mathbf{u}_1,\ldots,\mathbf{u}_{q^k})
    		\text{ and }
    		\mathbf{D}^{(2)}
    		\triangleq
    		\mathbf{D}_f^{(2)}(t,\mathbf{u}_1,\ldots,\mathbf{u}_{q^k})$. 
    		We first show that $r_p^f(k,t)\ge N_p(\mathbf{D}^{(1)})$. 
    		Suppose, to the contrary, that $r_p^f(k,t)<N_p(\mathbf{D}^{(1)})$. 
    		Let $\mathrm{Enc}:\mathbb{Z}_q^k \to \mathbb{Z}_q^{k+r},
    		\,
    		\mathbf{u}_i \mapsto (\mathbf{u}_i,\mathbf{p}_i)$, 
    		with \(r=r_p^f(k,t)\) that defines an FCSPLC for the function \(f\). Then, there exist \(i,j\in[q^k]\) with \(f(\mathbf{u}_i)\neq f(\mathbf{u}_j)\) such that $d_p(\mathbf{p}_i,\mathbf{p}_j)
    		<
    		[\mathbf{D}^{(1)}]_{ij}
    		=
    		\left[
    		\left(2t+1-\frac{q}{2}\right)
    		-d_p(\mathbf{u}_i,\mathbf{u}_j)
    		\right]^+$.
    		Since \(f(\mathbf{u}_i)\neq f(\mathbf{u}_j)\), it follows that $d_p(\mathbf{p}_i,\mathbf{p}_j)
    		<
    		\left(2t+1-\frac{q}{2}\right)
    		-d_p(\mathbf{u}_i,\mathbf{u}_j)$.
    		By Lemma~\ref{lem_delta},
    		\[
    		\begin{aligned}
    			d_p\bigl(\mathrm{Enc}(\mathbf{u}_i),
    			\mathrm{Enc}(\mathbf{u}_j)\bigr)
    			&\le
    			d_p(\mathbf{u}_i,\mathbf{u}_j)
    			+d_p(\mathbf{p}_i,\mathbf{p}_j)
    			+\frac{q}{2}\\
    			&<
    			d_p(\mathbf{u}_i,\mathbf{u}_j)
    			+\left(2t+1-\frac{q}{2}\right)
    			-d_p(\mathbf{u}_i,\mathbf{u}_j)
    			+\frac{q}{2}\\
    			&=
    			2t+1,
    		\end{aligned}
    		\]
    		which contradicts the definition of an FCSPLC. Therefore, $r_p^f(k,t)\ge N_p(\mathbf{D}^{(1)})$.
    		
    		Next, we show that $r_p^f(k,t)\le N_p(\mathbf{D}^{(2)})$.
    		Let $\mathcal{P}
    		=
    		\{\mathbf{p}_1,\ldots,\mathbf{p}_{q^k}\}$ 
    		be a \(\mathbf{D}^{(2)}_p\)-code of length \(N_p(\mathbf{D}^{(2)})\), and define $\mathrm{Enc}:\mathbb{Z}_q^k \to \mathbb{Z}_q^{k+r},
    		\,
    		\mathbf{u}_i \mapsto (\mathbf{u}_i,\mathbf{p}_i)$. For any \(i,j\in[q^k]\) satisfying
    		\(f(\mathbf{u}_i)\neq f(\mathbf{u}_j)\), we have $d_p(\mathbf{p}_i,\mathbf{p}_j)
    		\ge
    		[\mathbf{D}^{(2)}]_{ij}
    		=
    		\left(2t+1+\frac{q}{2}\right)
    		-d_p(\mathbf{u}_i,\mathbf{u}_j)$. 
    		Therefore, by Lemma~\ref{lem_delta},
    		\[
    		\begin{aligned}
    			d_p\bigl(\mathrm{Enc}(\mathbf{u}_i),
    			\mathrm{Enc}(\mathbf{u}_j)\bigr)
    			&\ge
    			d_p(\mathbf{u}_i,\mathbf{u}_j)
    			+d_p(\mathbf{p}_i,\mathbf{p}_j)
    			-\frac{q}{2}\\
    			&\ge
    			d_p(\mathbf{u}_i,\mathbf{u}_j)
    			+\left(2t+1+\frac{q}{2}\right)
    			-d_p(\mathbf{u}_i,\mathbf{u}_j)
    			-\frac{q}{2}\\
    			&=
    			2t+1.
    		\end{aligned}
    		\]
    		Hence, \(\mathrm{Enc}\) defines an FCSPLC for the function \(f\). Consequently, $r_p^f(k,t)\le N_p(\mathbf{D}^{(2)})$.
    \end{IEEEproof}
    \subsection{Simplified Redundancy Bounds}
    Theorem~\ref{Red_Bounds} provides general lower and upper bounds on the optimal redundancy \( r_{p}^f(q, k, t) \) in terms of $N_{p}\big (\boldsymbol {D}_{f}^{(1)}(t,\boldsymbol {u}_{1},\ldots,\boldsymbol {u}_{q^{k}})\big)$ and $N_{p}\big (\boldsymbol {D}_{f}^{(2)}(t,\boldsymbol {u}_{1},\ldots,\boldsymbol {u}_{q^{k}})\big)$. However, determining these minimum lengths exactly is generally difficult. For ease of evaluation, we derive simplified and possibly sub-optimal bounds by considering subsets of $\mathbb{Z}_q^k$. The following Corollary presents a simplified lower bound on \( r_{p}^f(q, k, t) \).
    \begin{cor} \label{cor1_LB}
    	Let $\{\mathbf{u}_1, \ldots, \mathbf{u}_M\}$ be an arbitrary subset of $\mathbb{Z}_q^k$. Then the optimal redundancy of an FCSPLC for $f$ satisfies $r_{p}^f(q, k, t) \ge N_{p}\big (\boldsymbol {D}_{f}^{(1)}(t,\boldsymbol {u}_{1},\ldots,\boldsymbol {u}_{M})\big).$ For any function $f$ with $|Im(f) \ge 2$, we have $r_{p}^f(q, k, t) \geq \left\lceil \frac{2t-1-\frac{q}{2}}{\frac{q}{2}} \right\rceil$.
    \end{cor}
    \begin{IEEEproof}
    	Since $\{\mathbf{u}_1, \ldots, \mathbf{u}_M\}$ is a subset of $\mathbb{Z}_q^k$, it follows that $\{\mathbf{p}_1, \ldots, \mathbf{p}_M\}$ is a ${D}_{f}^{(1)}(t,\boldsymbol {u}_{1},\ldots,\boldsymbol {u}_{M})$-code, and  $$N_{p}\big (\boldsymbol {D}_{f}^{(1)}(t,\boldsymbol {u}_{1},\ldots,\boldsymbol {u}_{q^k})\big) \ge N_{p}\big (\boldsymbol {D}_{f}^{(1)}(t,\boldsymbol {u}_{1},\ldots,\boldsymbol {u}_{M})\big).$$ By Theorem~\ref{Red_Bounds}, $r_{p}^f(q, k, t) \ge N_{p}\big (\boldsymbol {D}_{f}^{(1)}(t,\boldsymbol {u}_{1},\ldots,\boldsymbol {u}_{q^k})\big)$. Hence, $r_{p}^f(q, k, t) \ge N_{p}\big (\boldsymbol {D}_{f}^{(1)}(t,\boldsymbol {u}_{1},\ldots,\boldsymbol {u}_{M})\big).$ For $|Im(f) \ge 2$, there exists $\mathbf{u}, \mathbf{v} \in \mathbb Z_q^k$ such that $f(\mathbf{u}) \ne f(\mathbf{v})$ and $d_p(\mathbf{u}, \mathbf{v})=2$. Thus, $$r_{p}^f(q, k, t) \ge N_{p}\big (\boldsymbol {D}_{f}^{(1)}(t,\boldsymbol {u}_{1},\ldots,\boldsymbol {u}_{M})\big)=N_{p}(2,2t-1-\frac q2)=\left\lceil \frac{2t-1-\frac{q}{2}}{\frac{q}{2}} \right\rceil,$$ where the last equality is obtained by the repetition code $\{(0, \ldots, 0), (\frac q2, \ldots, \frac q2)\}$ of length $\left\lceil \frac{2t-1-\frac{q}{2}}{\frac{q}{2}} \right\rceil$.
    \end{IEEEproof}
    We next derive a simplified upper bound on the optimal redundancy \( r_{p}^f(q, k, t) \) by considering a representative subset of information vectors with distinct function values. To this end, we consider the following definitions.
    \begin{defn}
    	The pair Lee-distance between two function values $f_1, f_2 \in Im(f) $ is defined as the smallest pair Lee-distance between two information vectors that evaluate to $f_1$ and $f_2$, i.e., 
        \begin{equation*} d_{p}\left ({{f_{1}, f_{2}}}\right) \!\triangleq \!\min _{\boldsymbol {u}_{1}, \boldsymbol {u}_{2} \in \mathbb {Z}_{q}^{k}} d_{p}\left ({{\boldsymbol {u}_{1}, \boldsymbol {u}_{2}}}\right) \text { s.t.}~ f\left ({{\boldsymbol {u}_{1}}}\right)\!=\!f_{1}~\text {and}~f\left ({{\boldsymbol {u}_{2}}}\right)\!=\!f_{2}.\end{equation*}
    \end{defn}
    \begin{defn}
    	The function pair Lee-distance matrices of a function $f$ are denoted by $\mathbf{D}_f^{(1)}(t,f_1,\ldots,f_E)
    	 \\ \text{and } 
    	\mathbf{D}_f^{(2)}(t,f_1,\ldots,f_E)$,
    	where $E=|\operatorname{Im}(f)|$ and
    	$\operatorname{Im}(f)=\{f_1,\ldots,f_E\}$.
    	The entries of these $E\times E$ matrices are given by
    	\[
    	\bigl[\mathbf{D}_f^{(1)}(t,f_1,\ldots,f_E)\bigr]_{ij}
    	=
    	\begin{cases}
    		\left[
    		\left(2t+1-\dfrac{q}{2}\right)
    		-d_p(f_i,f_j)
    		\right]^+,
    		& \text{if } i\neq j,\\[1ex]
    		0,
    		& \text{otherwise},
    	\end{cases}
    	\]
    	and
    	\[
    	\bigl[\mathbf{D}_f^{(2)}(t,f_1,\ldots,f_E)\bigr]_{ij}
    	=
    	\begin{cases}
    		\left[
    		\left(2t+1+\dfrac{q}{2}\right)
    		-d_p(f_i,f_j)
    		\right]^+,
    		& \text{if } i\neq j,\\[1ex]
    		0,
    		& \text{otherwise}.
    	\end{cases}
    	\]
    \end{defn} With these definitions, we derive a simplified upper bound on the optimal redundancy.
    \begin{thm} \label{red_fdm}
         For any function $f:\mathbb{Z}_q^k \to \operatorname{Im}(f)$, we have $r_p^f(k,t)
    		\le
    		N_p\bigl(\mathbf{D}_f^{(2)}(t,f_1,\ldots,f_E)\bigr)$, 
    		where $\operatorname{Im}(f)=\{f_1,\ldots,f_E\}$.
    	\end{thm}
    	
    	\begin{IEEEproof}
    	Let $\mathcal{P}
    		=
    		\{\mathbf{p}_1,\ldots,\mathbf{p}_E\}$ 
    		be a
    		$\mathbf{D}_f^{(2)}(t,f_1,\ldots,f_E)_p$-code
    		of length
    		$N_p(\mathbf{D}_f^{(2)}(t,f_1,\ldots,f_E))$. Assign the redundancy vector $\mathbf{p}_i$ to every information vector
    		$\mathbf{u}\in\mathbb{Z}_q^k$ satisfying
    		$f(\mathbf{u})=f_i$.
    		Thus, all information vectors having the same function value are assigned the same redundancy vector. Define the encoding function $\mathrm{Enc}:\mathbb{Z}_q^k
    		\to
    		\mathbb{Z}_q^{k+r},
    		\,
    		\mathrm{Enc}(\mathbf{u})
    		=
    		(\mathbf{u},\mathbf{p}_i)$, 
    		where $f(\mathbf{u})=f_i$. Since $\mathcal{P}$ is a
    		$\mathbf{D}_f^{(2)}(t,f_1,\ldots,f_E)_p$-code,
    		for any distinct function values $f_i,f_j$, $d_p(\mathbf{p}_i,\mathbf{p}_j)
    		\ge
    		\left[
    		\left(2t+1+\frac{q}{2}\right)
    		-
    		d_p(f_i,f_j)
    		\right]^+$. Now let $\mathbf{u},\mathbf{v}\in\mathbb{Z}_q^k$ satisfy
    		$f(\mathbf{u})=f_i$,
    		$f(\mathbf{v})=f_j$,
    		and $f_i\neq f_j$.
    		By the definition of $d_p(f_i,f_j)$, $d_p(\mathbf{u},\mathbf{v})
    		\ge
    		d_p(f_i,f_j)$.
    		Combining this with the previous inequality gives $d_p(\mathbf{u},\mathbf{v})
    		+
    		d_p(\mathbf{p}_i,\mathbf{p}_j)
    		\ge
    		2t+1+\frac{q}{2}$. Using Lemma~\ref{lem_delta}, we obtain
    		$d_p\bigl(\mathrm{Enc}(\mathbf{u}),
    			\mathrm{Enc}(\mathbf{v})\bigr)
    			\ge
    			d_p(\mathbf{u},\mathbf{v})
    			+
    			d_p(\mathbf{p}_i,\mathbf{p}_j)
    			-
    			\frac{q}{2}
    			\ge
    			2t+1.$
    		Hence, the encoding function $\mathrm{Enc}$ defines an FCSPLC for the function $f$. Therefore, an FCSPLC exists with redundancy
    		$N_p(\mathbf{D}_f^{(2)}(t,f_1,\ldots,f_E))$, implying that
    		$r_p^f(k,t)
    		\le
    		N_p\bigl(\mathbf{D}_f^{(2)}(t,f_1,\ldots,f_E)\bigr)$.
    	\end{IEEEproof}
   
    	Since the sizes of $\boldsymbol {D}_{f}^{(1)}(t,\boldsymbol {u}_{1},\ldots,\boldsymbol {u}_{M})$ and $\mathbf{D}_f^{(2)}(t,f_1,\ldots,f_E)$ are are typically much smaller than $\boldsymbol {D}_{f}^{(1)}(t,\boldsymbol {u}_{1},\ldots,\boldsymbol {u}_{q^k})$ and $\boldsymbol {D}_{f}^{(2)}(t,\boldsymbol {u}_{1},\ldots,\boldsymbol {u}_{q^k})$, respectively, the corresponding minimum lengths $N_{p}\big (\boldsymbol {D}_{f}^{(1)}\\(t,\boldsymbol {u}_{1},\ldots,\boldsymbol {u}_{M})\big)$ and $N_p\bigl(\mathbf{D}_f^{(2)}(t,f_1,\ldots,f_E)\bigr)$ are generally easier to evaluate.
    \subsection{Bounds on the Shortest Length of Irregular-Pair Lee-distance Codes}
    To determine the optimal redundancy, it is necessary to evaluate $N_p(\mathbf{D})$. While the exact value of $N_p(\mathbf{D})$ is generally difficult to obtain, useful lower and upper bounds can be derived. 
    The following lemmas are used in the derivation of a Plotkin-type lower bound on the shortest length of irregular-pair Lee-distance  codes over $\mathbb{Z}_q$ in Theorem~\ref{Plotkin_bound}. 
    \begin{lem} \label{Lem1_PB}
    	Let $S_k=\frac12\sum_{a,b \in (\mathbb{Z}_q \times \mathbb{Z}_q)} n_a n_b d_p(a,b)$, where $d_p(a,b)\le \frac q2 \text{ for all }  a,b\in (\mathbb{Z}_q \times \mathbb{Z}_q)$, and $\sum_{a\in (\mathbb{Z}_q \times \mathbb{Z}_q)} n_a=M$. Then $S_k$ is maximized when $\mathcal{A}=\{a\in (\mathbb{Z}_q \times \mathbb{Z}_q):n_a>0 \}$ is a largest subset satisfying $d_p(a, b) =\frac{q}{2}, \forall \, a \neq b \in \mathcal{A}$, and distributing $M$ as uniformly as possible among the elements of $\mathcal{A}$, i.e., $|n_a-n_b| \leq 1 \, \forall \, a, b \in \mathcal{A}$. Equivalently, if $L=|\mathcal{A}|$, then $n_a\in \{\left\lfloor \frac{M}{L}\right\rfloor, \left\lceil \frac{M}{L}\right\rceil\}$.
    \end{lem}
    \begin{IEEEproof}
    	Since $d_p(a,a)=0$ for all $a\in (\mathbb{Z}_q\times \mathbb{Z}_q)$, we write
    	\begin{equation}
    		S_k=\frac12\sum_{\substack{a\neq b\in (\mathbb{Z}_q\times \mathbb{Z}_q) }} n_a n_b d_p(a,b).
    	\end{equation} 
    	Let $\mathcal A=\{a\in(\mathbb Z_q\times\mathbb Z_q): n_a>0\}$. Since $d_p(a,b)\le \frac q2$ for all $a,b\in(\mathbb Z_q\times\mathbb Z_q)$, it follows that
    	\begin{equation}\label{eq2}
    		S_k
    		\le
    		\frac12
    		\sum_{a\neq b}
    		n_a n_b
    		\frac q2
    		=
    		\frac q4
    		\sum_{a\neq b}
    		n_a n_b.
    	\end{equation}
    	Equality holds in (\ref{eq2}) if and only if $d_p(a,b)=\frac q2$ for every pair $a\neq b$ satisfying $n_a n_b>0$.
    	Since $\mathcal A=\{a\in(\mathbb Z_q\times\mathbb Z_q):n_a>0\}$, the above condition can be rewritten as $d_p(a,b)=\frac q2,\,\forall\, a\neq b\in\mathcal A$. Therefore, for fixed values of $n_a$, the maximum possible value of $S_k$ is attained precisely when all elements of $\mathcal A$ are mutually at distance $\frac q2$.
    	
    	Assume now that $\mathcal{A}$ is such an equidistant subset and let $L=|\mathcal{A}|$. Then
    	\begin{equation}
    		S_k=\frac{q}{4}\sum_{a \neq b} n_a n_b.
    	\end{equation}
    	Using $\sum_{a\neq b} n_a n_b=\left(\sum_{a\in\mathcal{A}} n_a\right)^2 -\sum_{a\in\mathcal{A}} n_a^2$, together with $\sum_{a\in\mathcal{A}} n_a=M$, we obtain 
    	\begin{equation} \label{eq4}
    		S_k=\frac{q}{4}
    		\left(
    		M^2-\sum_{a\in\mathcal{A}} n_a^2
    		\right).
    	\end{equation}
    	Thus maximizing $S_k$ is equivalent to minimizing $\sum_{a\in\mathcal{A}} n_a^2$. By the Cauchy--Schwarz inequality,
    	\begin{equation} \label{eq5}
    		\sum_{a\in\mathcal{A}} n_a^2
    		\ge
    		\frac{1}{L}
    		\left(
    		\sum_{a\in\mathcal{A}} n_a
    		\right)^2=\frac{M^2}{L}.
    	\end{equation}
    	Since the $n_a$ are integers, the minimum of $\sum_{a\in\mathcal{A}} n_a^2$ subject to $\sum_{a\in\mathcal{A}} n_a=M$
    	is attained when the $n_a$ are as equal as possible. Indeed, suppose there exist $a,b\in\mathcal{A}$ such that $n_a-n_b\ge 2$. Replacing $n_a$ by $n_a-1$ and $n_b$ by $n_b+1$ preserves the constraint $\sum_{a\in\mathcal{A}} n_a=M$ and changes the sum of squares by $(n_a-1)^2+(n_b+1)^2-(n_a^2+n_b^2)=-2(n_a-n_b-1)<0$. Thus, the sum of squares is strictly decreased, contradicting its minimality. Hence, at the minimum, $|n_a-n_b|\le 1$ for all $a,b\in\mathcal{A}$. Equivalently, $n_a\in \left\{\left\lfloor\frac{M}{L}\right\rfloor, \left\lceil\frac{M}{L}\right\rceil \right\},\, \forall\, a\in\mathcal{A}$. Substituting (\ref{eq5}) into (\ref{eq4}) yields
    	\begin{equation}
    		S_k
    		\le
    		\frac{q}{4}M^2
    		\left(
    		1-\frac{1}{L}
    		\right).
    	\end{equation}
    	
    	Finally, observe that the function $1-\frac{1}{L}$ is strictly increasing in $L$. Hence the upper bound on $S_k$ increases with the size of the equidistant subset $\mathcal A$. Therefore, $S_k$ is maximized when $\mathcal A$ is a largest subset of $(\mathbb Z_q\times\mathbb Z_q)$ satisfying $d_p(a,b)=\frac q2$ for all $a\neq b\in\mathcal A$, and the values $n_a$ are distributed as uniformly as possible over $\mathcal A$, i.e., $|n_a-n_b|\le 1$ for all $a,b\in\mathcal A$. 
    \end{IEEEproof}
    \begin{lem}\label{Lem2_PB}
    	Let $\mathcal{A}=\{x, x+\frac q 2\}\times\{y, y+\frac q 2\}, x,y\in\mathbb Z_q$. Then $\mathcal{A}$ is a maximum-cardinality subset of $(\mathbb Z_q \times \mathbb Z_q)$ such that $d_p(a,b)=\frac q 2, \forall \, a\neq b\in A$ with $|\mathcal{A}|=4$.
    \end{lem}
    \begin{IEEEproof}
    	Let $\mathcal A=\{x,x+\frac q 2\}\times\{y,y+\frac q 2\}$. Explicitly, $\mathcal A= \{(x,y),(x,y+\frac q 2), (x+\frac q 2,y),(x+\frac q 2,y+\frac q 2)\}$.
    	Hence, $|\mathcal A|=2\cdot 2=4$. Now let $a=(a_1,a_2)$ and $b=(b_1,b_2)$ be distinct elements of $\mathcal A$. For each coordinate $i\in\{1,2\}$, either $a_i=b_i$ or $a_i-b_i\equiv \frac q 2 \pmod q$. Therefore, $d_L(a_i,b_i)\in\{0,\frac q 2\}$. Since $a\neq b$, at least one coordinate differs, and hence $d_L(a_i,b_i)=\frac q 2$
    	for some $i$. Consequently, $d_p(a,b)=\max\{d_L(a_1,b_1),d_L(a_2,b_2)\}=\frac q2$. Thus every distinct pair of elements of $\mathcal A$ is at symbol-pair Lee-distance $\frac q2$.
    	
    	We already know that $\mathcal A=\{x,x+\tfrac q2\}\times\{y,y+\tfrac q2\}$ has cardinality \(4\) and satisfies $d_p(a,b)=\frac q2,\, \forall\, a\neq b\in\mathcal A$. It remains to show that no larger subset of \(\mathbb Z_q\times\mathbb Z_q\) can have this property. Let \(\mathcal S\subseteq \mathbb Z_q\times\mathbb Z_q\) be such that $d_p(u,v)=\frac q2,\, \forall\, u\neq v\in\mathcal S$. Choose an arbitrary point \((x,y)\in\mathcal S\). Since the maximum Lee-distance in \(\mathbb Z_q\) is \(\tfrac q2\), every other point of \(\mathcal S\) must differ from \((x,y)\) by \(\tfrac q2\) in at least one coordinate. Hence $\mathcal S\setminus\{(x,y)\} \subseteq V\cup H$, where $V=\{(x+\tfrac q2,t):t\in\mathbb Z_q\},\, H=\{(s,y+\tfrac q2):s\in\mathbb Z_q\}$. 
    	
    	First consider \(V\). If \((x+\tfrac q2,t_1)\) and \((x+\tfrac q2,t_2)\) are distinct points of \(\mathcal S\), then their first coordinates coincide, so $d_p\big((x+\tfrac q2,t_1),(x+\tfrac q2,t_2)\big)=d_L(t_1,t_2)$. Therefore \(d_L(t_1,t_2)=\tfrac q2\). Since $q$ is even, the unique element of $\mathbb{Z}_q$ at Lee-distance $\tfrac q2$ from $t_1$ is $t_1+\tfrac q2$, it follows that \(t_2=t_1+\tfrac q2\). Hence \(V\cap\mathcal S\) contains at most two points, i.e., $|V\cap\mathcal S|\le 2$ . Similarly, $|H\cap\mathcal S|\le 2$. If $|V\cap\mathcal S|\le 1$ then $|S| \le 1+1+2=4$. Therefore, it remains to consider the case \(|V\cap\mathcal S|=2\). 
    	
    	Suppose \(|V\cap\mathcal S|=2\). Write $V\cap\mathcal S=\Big\{\big(x+\tfrac q2,t\big),\big(x+\tfrac q2,t+\tfrac q2\big)
    	\Big\}$. Let \((s,y+\tfrac q2)\in H\cap\mathcal S\). Since it must be at a pair-distance of \(\tfrac q2\) from both the points above, we have $d_p\!\left(\big(x+\tfrac q2,t\big),(s,y+\tfrac q2)\right)=\frac q2$ and $d_p\!\left(\big(x+\tfrac q2,t+\tfrac q2\big),(s,y+\tfrac q2)\right)=\frac q2$. For the first equality, either \(s=x\) or \(t=y\). For the second equality, either \(s=x\) or \(t=y+\tfrac q2\). If \(s\neq x\), then both \(t=y\) and \(t=y+\tfrac q2\) must hold,
    	which is impossible. Therefore \(s=x\). Hence every point of \(H\cap\mathcal S\) must equal $(x,y+\tfrac q2)$, so
    	$|H\cap\mathcal S|\le 1$. Consequently, $|\mathcal S| \le 1+2+1=4$. Therefore every subset of \(\mathbb Z_q\times\mathbb Z_q\) whose distinct elements are pairwise at distance \(\frac q2\) has cardinality at most \(4\). Since \(\mathcal A\) has cardinality \(4\), it is a maximum-cardinality subset. 
    \end{IEEEproof} 
    The following theorem provides a Plotkin-type lower bound on the shortest length of irregular-pair Lee-distance codes over $\mathbb{Z}_q$.
    \begin{thm} \label{Plotkin_bound}
    For any distance matrix $\mathbf{D} \in \mathbb{N}_{0}^{M \times M}$,
    \begin{align} N_{p}(\boldsymbol {D})\geq 
    	\begin{cases}
    		\frac {16}{3M^{2}q}\sum \limits _{i,j,i< j}[\boldsymbol {D}]_{ij}, & M\equiv 0\pmod{4}, \\ \frac {16}{(3M^{2}-4)q}\sum \limits _{i,j,i< j}[\boldsymbol {D}]_{ij}, & M\equiv 2\pmod{4}, \\ \frac {16}{3(M^{2}-1)q}\sum \limits _{i,j,i< j}[\boldsymbol {D}]_{ij}, & M \text { is odd}.
    	\end{cases}
    \end{align}	
    \end{thm}
    \begin{IEEEproof}
    	Let $\mathcal P=\{\mathbf p_1,\mathbf p_2,\ldots,\mathbf p_M\}\subseteq \mathbb Z_q^r$ be a $\mathbf D_p$-code corresponding to the DRM $\mathbf D$, and let $N_p(\mathbf D)=r$. By Definition of $\mathbf D_p$ codes,
    	$d_{p}(\mathbf{p}_i, \mathbf{p}_j) \geq [\mathbf{D}]_{ij}, \text{ for all } i, j \in [M]$.
    	Therefore,
    	\begin{equation}\label{eq:sum_D}
    		\sum_{i,j \in [M] : i <j} [\mathbf{D}]_{ij} \le \sum_{i,j \in [M] : i <j} d_p(\mathbf{p}_i, \mathbf{p}_j).
    	\end{equation} 
    	
    	Let $S=\sum_{i<j} d_p(\mathbf p_i,\mathbf p_j)$. For each coordinate $k\in[r]$, let $S_k=\frac12
    	\sum_{a,b\in(\mathbb Z_q\times\mathbb Z_q)}
    	n_a^{(k)}n_b^{(k)}d_p(a,b)$,
    	where $n_a^{(k)}$ denotes the number of occurrences of the pair
    	$a\in(\mathbb Z_q\times\mathbb Z_q)$ in the $k$th coordinate column of the codewords. Then $S=\sum_{k=1}^{r} S_k$. For notational simplicity, we suppress the superscript $(k)$ from $n_a^{(k)}$.
    	
    	By Lemmas~\ref{Lem1_PB} and~\ref{Lem2_PB}, $S_k$ is maximized when
    	the integers $n_a$ are distributed as uniformly as possible, i.e.,
    	$|n_a-n_b|\le 1$ for all $a,b\in\mathcal A$, subject to
    	$\sum_{a\in\mathcal A}n_a=M$, and $|\mathcal A|=4$.
    	Hence
    	\begin{equation}
    		S_k\le \frac q4\left(M^2-\sum_{a\in\mathcal A}n_a^2\right).
    	\end{equation}
    	
    	Since the same optimization arises for every coordinate
    	column, the above upper bound is valid for each $S_k$. Therefore,
    	\[
    	S=\sum_{k=1}^{r}S_k
    	\le
    	\sum_{k=1}^{r}
    	\frac q4\left(M^2-\sum_{a\in\mathcal A}n_a^2\right)
    	=
    	\frac q4\left(M^2-\sum_{a\in\mathcal A}n_a^2\right)r.
    	\]
    	Combining this with (\ref{eq:sum_D}) yields
    	\begin{equation}\label{eq:final_bound}
    		\sum_{i,j\in[M]:\,i<j}[\mathbf D]_{ij}
    		\le
    		\frac q4\left(M^2-\sum_{a\in\mathcal A}n_a^2\right)r.
    	\end{equation}
    	We now evaluate the maximum possible value of
    	$M^2-\sum_{a\in\mathcal A}n_a^2$
    	for the different residue classes of $M$ modulo $4$.
    	
    	\textit{Case 1}: $M\equiv0\pmod4:$
    	In this case, $n_a=\frac{M}{4}$ for all $a\in\mathcal A$. Hence $\sum_{a\in\mathcal A}n_a^2
    	=
    	4\left(\frac{M}{4}\right)^2
    	=
    	\frac{M^2}{4}$.
    	Substituting into (\ref{eq:final_bound}) gives $\sum_{i<j}[\mathbf D]_{ij}
    	\le
    	\frac{3qM^2}{16}\,r$,
    	and therefore $$r
    	\ge
    	\frac{16}{3M^2q}
    	\sum_{i<j}[\mathbf D]_{ij}.$$

    	\textit{Case 2}: $M\equiv1\pmod4:$
    	For this case, $n_a\in
    	\left\{
    	\frac{M-1}{4},
    	\frac{M-1}{4},
    	\frac{M-1}{4},
    	\frac{M+3}{4}
    	\right\}$.
    	Thus $\sum_{a\in\mathcal A}n_a^2
    	=
    	\frac{M^2+3}{4}$,
    	and hence $M^2-\sum_{a\in\mathcal A}n_a^2
    	=
    	\frac{3(M^2-1)}{4}$.
    	Substituting into (\ref{eq:final_bound}) yields $\sum_{i<j}[\mathbf D]_{ij}
    	\le
    	\frac{3q(M^2-1)}{16}\,r$,
    	which implies $$r
    	\ge
    	\frac{16}{3(M^2-1)q}
    	\sum_{i<j}[\mathbf D]_{ij}.$$
    	
    	\textit{Case 3}: $M\equiv2\pmod4:$
    	In this case, $n_a\in
    	\left\{
    	\frac{M-2}{4},
    	\frac{M-2}{4},
    	\frac{M+2}{4},
    	\frac{M+2}{4}
    	\right\}$.
    	Therefore, $\sum_{a\in\mathcal A}n_a^2
    	=
    	\frac{M^2+4}{4},
    	\,
    	M^2-\sum_{a\in\mathcal A}n_a^2
    	=
    	\frac{3M^2-4}{4}$.
    	Substituting into (\ref{eq:final_bound}) gives $\sum_{i<j}[\mathbf D]_{ij}
    	\le
    	\frac{(3M^2-4)q}{16}\,r$,
    	and hence $$r
    	\ge
    	\frac{16}{(3M^2-4)q}
    	\sum_{i<j}[\mathbf D]_{ij}.$$
    	
    	\textit{Case 4}: $M\equiv3\pmod4:$
    	For this case, $n_a\in
    	\left\{
    	\frac{M-3}{4},
    	\frac{M+1}{4},
    	\frac{M+1}{4},
    	\frac{M+1}{4}
    	\right\}$.
    	Thus $\sum_{a\in\mathcal A}n_a^2
    	=
    	\frac{M^2+3}{4},
    	\,
    	M^2-\sum_{a\in\mathcal A}n_a^2
    	=
    	\frac{3(M^2-1)}{4}$.
    	Substituting into (\ref{eq:final_bound}) yields $\sum_{i<j}[\mathbf D]_{ij}
    	\le
    	\frac{3q(M^2-1)}{16}\,r$,
    	and therefore $$r
    	\ge
    	\frac{16}{3(M^2-1)q}
    	\sum_{i<j}[\mathbf D]_{ij}.$$
    	Since $r=N_p(\mathbf D)$, the stated bound follows.
    	\begin{align*} N_{p}(\boldsymbol {D})\geq 
    		\begin{cases}
    			\frac {16}{3M^{2}q}\sum \limits _{i,j,i< j}[\boldsymbol {D}]_{ij}, & M\equiv 0\pmod{4}, \\ \frac {16}{(3M^{2}-4)q}\sum \limits _{i,j,i< j}[\boldsymbol {D}]_{ij}, & M\equiv 2\pmod{4}, \\ \frac {16}{3(M^{2}-1)q}\sum \limits _{i,j,i< j}[\boldsymbol {D}]_{ij}, & M \text { is odd}.
    		\end{cases}
    	\end{align*}
    \end{IEEEproof}
    As a special case, for $q=2$, the proposed bound reduces to the Plotkin-type bound for function-correcting symbol-pair codes (FCSPCs) established in \cite{XLC}, given by 
    \begin{align*} N_{p}(\boldsymbol {D})\geq 
    	\begin{cases}
    		\frac {8}{3M^{2}}\sum \limits _{i,j,i< j}[\boldsymbol {D}]_{ij}, & M\equiv 0\pmod{4}, \\ \frac {8}{(3M^{2}-4)}\sum \limits _{i,j,i< j}[\boldsymbol {D}]_{ij}, & M\equiv 2\pmod{4}, \\ \frac {8}{3(M^{2}-1)}\sum \limits _{i,j,i< j}[\boldsymbol {D}]_{ij}, & M \text { is odd}.
    	\end{cases}
    \end{align*}
    \begin{rem}
    	For bijective functions, we have $[\boldsymbol{D}]_{ij}=d$ for all $i\neq j$, and hence $\sum_{i<j}[\boldsymbol{D}]_{ij}
    	=
    	\frac{M(M-1)}{2}d$.
    	Substituting this into the bound of Theorem~\ref{Plotkin_bound} gives
    	{\small
    		\[
    		N_{p}(M, d)\ge
    		\begin{cases}
    			\dfrac{8(M-1)d}{3Mq}, & M\equiv 0\pmod{4},\\[1ex]
    			\dfrac{8M(M-1)d}{(3M^{2}-4)q}, & M\equiv 2\pmod{4},\\[1ex]
    			\dfrac{8M(M-1)d}{3(M^{2}-1)q}, & M \text{ is odd},
    		\end{cases}
    		\]
    	}a variant of Plotkin's bound for classical Lee-distance codes for symbol-pair channels with minimum pair-distance $d$ over $\mathbb Z_q$. 
    \end{rem}
    In addition to the above lower bound, we derive an upper bound on $N_p(\boldsymbol {D})$.
    \begin{lem}
    	Let $\boldsymbol{D}\in \mathbb{N}_0^{M\times M}$ be a distance matrix and let
    	$\pi:[M]\to[M]$ be any permutation. Then
    	\[
    	N_p(\boldsymbol{D})
    	\le
    	\min_{r\in\mathbb{N}}
    	\left\{
    	r:
    	q^r >
    	\max_{j\in[M]}
    	\sum_{i=1}^{j-1}
    	V_p^{L}\!\left(r,[\boldsymbol{D}]_{\pi(i)\pi(j)}-1\right)
    	\right\},
    	\]
    	where $V_p^{L}(r,t)$ denotes the number of vectors in
    	$\mathbb{Z}_q^r$ whose symbol-pair Lee-distance from a fixed vector
    	$\boldsymbol{x}\in\mathbb{Z}_q^r$ is at most $t$.
    \end{lem}
    
    \begin{proof}
    	We construct a symbol-pair Lee-distance code of length $r$
    	by iteratively selecting codewords satisfying the prescribed
    	distance constraints. For simplicity, assume first that
    	$\pi(i)=i$. Choose an arbitrary codeword
    	$\boldsymbol{p}_1\in\mathbb{Z}_q^r$.
    	Next, select a codeword $\boldsymbol{p}_2$ such that $d_{p}(\boldsymbol{p}_1,\boldsymbol{p}_2)
    	\ge [\boldsymbol{D}]_{12}$.
    	Such a codeword exists whenever $q^r >
    	V_p^{L}(r,[\boldsymbol{D}]_{12}-1)$.
    	
    	Suppose that codewords
    	$\boldsymbol{p}_1,\ldots,\boldsymbol{p}_{j-1}$
    	have already been selected so that $d_{p}(\boldsymbol{p}_i,\boldsymbol{p}_\ell)
    	\ge [\boldsymbol{D}]_{i\ell},
    	\,
    	1\le i<\ell\le j-1$.
    	To select $\boldsymbol{p}_j$, we require $d_{p}(\boldsymbol{p}_i,\boldsymbol{p}_j)
    	\ge [\boldsymbol{D}]_{ij},
    	\,
    	1\le i\le j-1$.
    	The number of vectors $\boldsymbol{x}\in\mathbb{Z}_q^r$ that fail to satisfy at least one of the distance requirements $d_{p}(\boldsymbol{p}_i,\boldsymbol{x})
    	\ge
    	[\boldsymbol{D}]_{ij},
    	\,
    	1\le i\le j-1$, 
    	is at most $\sum_{i=1}^{j-1}
    	V_p^{L}\!\left(r,[\boldsymbol{D}]_{ij}-1\right)$.
    	Hence, if $q^r >
    	\sum_{i=1}^{j-1}
    	V_p^{L}(r,[\boldsymbol{D}]_{ij}-1)$, 
    	there exists a valid choice of $\boldsymbol{p}_j$.	
    	Proceeding inductively, all $M$ codewords can be selected provided $q^r >
    	\max_{j\in[M]}
    	\sum_{i=1}^{j-1}
    	V_p^{L}(r,[\boldsymbol{D}]_{ij}-1)$.
    	Since the codewords may be chosen in any order, the same argument
    	holds for any permutation $\pi$ of $[M]$, proving the result.
    \end{proof}
    
    \begin{rem}
    	For bijective functions, setting $[\boldsymbol{D}]_{ij}=d$ for all $i\neq j$, we obtain
    	\[
    	N_p(M, d)
    	\le
    	\min_{r\in\mathbb{N}}
    	\left\{
    	r:
    	q^r>(M-1)V_p^{L}(r,d-1)
    	\right\}.
    	\]
    	a variant of Gilbert-Varshamov's bound for Lee-distance codes for symbol-pair read channels with minimum pair-distance $d$ over $\mathbb{Z}_q$.
    \end{rem} 
    The following lemma establishes a connection between the shortest lengths of symbol-pair Lee-distance codes and classical Lee metric codes.
    \begin{lem} \label{pair-Lee}
    	For any integers $M\ge 2$ and $d\ge 1$, $N_p(M,d)\le N_L(M,d)$, 
    	where $N_p(M,d)$ and $N_L(M,d)$ denote the shortest lengths of Lee-distance symbol-pair codes and classical Lee metric codes over $\mathbb{Z}_q$, respectively, having $M$ codewords and minimum distance at least $d$.
    \end{lem}
    
    \begin{IEEEproof}
    	Let $r=N_L(M,d)$. Then there exists a code $\mathcal{C}=\{\mathbf{c}_1,\mathbf{c}_2,\ldots,\mathbf{c}_M\}
    	\subseteq \mathbb{Z}_q^r$
    	with minimum Lee-distance at least $d$, that is, $d_L(\mathbf{c}_i,\mathbf{c}_j)\ge d,
    	\, i\neq j$.
    	By Lemma~\ref{dpdl}, $d_p(\mathbf{c}_i,\mathbf{c}_j)
    	\ge
    	d_L(\mathbf{c}_i,\mathbf{c}_j)
    	\ge d,
    	\, i\neq j$.
    	Hence, $\mathcal{C}$ is also a Lee-distance symbol-pair code of length $r$ with minimum symbol-pair Lee-distance at least $d$. Therefore, $N_p(M,d)\le N_L(M,d)$.
    \end{IEEEproof}
    Next, using the relationship between $N_p(\mathbf{D})$ and $N_p(M,d)$, we obtain the following upper bound on $N_p(\mathbf{D})$.
    \begin{lem} \label{npd-npmd}
    	Let $\mathbf{D}\in \mathbb{N}_0^{M\times M}$ and define $D_{\max}=\max_{i,j\in[M]}[\mathbf{D}]_{ij}$.
    	If $D_{\max}\le d$, then $N_p(\mathbf{D})\le N_p(M,d)$.
    \end{lem}
    \begin{IEEEproof}
    	Let $r=N_p(M,d)$ and let $\mathcal{C}=\{\mathbf{c}_1,\ldots,\mathbf{c}_M\}
    	\subseteq \mathbb{Z}_q^r$
    	be a symbol-pair Lee-distance code of length $r$ satisfying $d_p(\mathbf{c}_i,\mathbf{c}_j)\ge d,
    	\, i\neq j$.
    	Since $D_{\max}\le d$, we have $d_p(\mathbf{c}_i,\mathbf{c}_j)
    	\ge d
    	\ge D_{\max}
    	\ge [\mathbf{D}]_{ij},
    	\, i\neq j$.
    	 Hence, $\mathcal{C}$ is a $\mathbf{D}_p$-code. Therefore, $N_p(\mathbf{D})\le N_p(M,d)$.
    \end{IEEEproof}
    As an immediate consequence of Lemmas~\ref{pair-Lee} and~\ref{npd-npmd}, we have $N_p(\mathbf{D})\le N_L(M,d),$
    for any $d$ satisfying $\max_{i,j\in[M]}[\mathbf{D}]_{ij}\le d.$
    Furthermore, by combining Lemma~\ref{pair-Lee} with a known upper bound on $N_L(M,d)$ from \cite{GA}, we obtain the following explicit upper bound on $N_p(M,d)$.
    
    \begin{cor}\label{NpMd-upper}
    	For any integers $M\ge 2$ and $q\ge 5$, $N_p(M,d)
    	\le
    	N_L(M,d)
    	\le
    	\left\lceil
    	\frac{\ln(M-1)+(d-1)}
    	{\ln q-\ln 3}
    	\right\rceil.$
    \end{cor}
    
    The above corollary immediately yields the following upper bound on $N_p(\mathbf{D})$.
    
    \begin{cor}\label{NpD-upper}
    	Let $\mathbf{D}\in\mathbb{N}_0^{M\times M}$ and define $D_{\max}=\max_{i,j\in[M]}[\mathbf{D}]_{ij}.$
    	If $D_{\max}\le d$, then for any $q\ge 5$, $N_p(\mathbf{D})
    	\le
    	\left\lceil
    	\frac{\ln(M-1)+(d-1)}
    	{\ln q-\ln 3}
    	\right\rceil.$
    \end{cor}
    Since the Lee metric coincides with the Hamming metric over $\mathbb{Z}_2$ and with the homogeneous metric over $\mathbb{Z}_4$, Lemma~\ref{pair-Lee} allows several known upper bounds on $N_L(M,d)$ from \cite{Lenz}, \cite{ZXZG}, and \cite{LL} to be converted directly into upper bounds on $N_p(M,d)$, leading to the following corollaries.
    
    \begin{cor}
    	Let $d\in\mathbb N$ such that there exists a Hadamard matrix
    	of order $d$, and $M\le 4d$ over $\mathbb Z_2$. Then $N_p(M,d)\le 2d$.
    \end{cor}
    
    \begin{cor}
    	For any $M,d\in\mathbb N$ with $d\ge 10$ and $M\le d^2$ over
    	$\mathbb Z_2$, $N_p(M,d)
    	\le
    	\left\lceil
    	\frac{2d-2}
    	{1-2\sqrt{\ln(d)/d}}
    	\right\rceil.$
    \end{cor}
    
    \begin{cor}
    	For any $M,d\in\mathbb N$ with $M>1$ over $\mathbb Z_4$, $N_p(M,d)
    	\le
    	\left\lceil
    	\frac{\ln(M-1)+d-1}
    	{1-\ln 2}
    	\right\rceil.$
    \end{cor}
\section{Construction of FCSPLCs for Pair-Locally Bounded, Pair-Lee Weight, and Pair-Lee Weight Distribution Functions} \label{section5}
In this section, we present explicit constructions of FCSPLCs for pair-locally bounded functions, the pair-Lee weight function, and the pair-Lee weight distribution function. In each case, the proposed construction yields FCSPLCs and establishes an upper bound on the optimal redundancy.
\subsection{Pair-Locally Bounded Functions}
The class of locally bounded functions, which take only a limited number of distinct values within any Lee ball, was studied in \cite{GA} and \cite{HRS}. We extend this notion to symbol-pair read channels and study FCSPLCs for the resulting class, referred to as \emph{pair-locally bounded functions}. 
\begin{defn}[\!\!\cite{GA},\cite{HRS}]
	The function ball of a function $f : \mathbb{Z}_q^k \to Im(f)$ with pair-radius $\rho$ around $u$ is defined as $B_p^f(\mathbf{u}, \rho) = \{ f(\mathbf{v})\;|\; \mathbf{v} \in \mathbb{Z}_q^k \text{ and } d_p(\mathbf{u}, \mathbf{v}) \le \rho \}.$
\end{defn}
\begin{defn}[\!\!\cite{GA}, \cite{HRS}]
	A function $f:\mathbb{Z}_q^k\to Im(f)$ is said to be a pair-locally $(\rho,\lambda)$-bounded function if $\mid B_p^f(\mathbf{u},\rho)\mid\leq \lambda$, $\forall \mathbf{u}\in \mathbb{Z}_q^k$.
\end{defn}
The following lemma is used in the construction of FCSPLCs for pair-locally bounded functions.
\begin{lem}[\!\!\cite{GA}]
	\label{Cont_LB}
	Let $f:\mathbb{Z}_q^k\to Im(f)$ be a pair-locally $(\rho,\lambda)$-bounded function. If there is total order on the image set $Im(f)$ such that  $B_p^f(\mathbf{u}, \rho)$ form a contiguous block for all $\mathbf{u}\in \mathbb{Z}_q^k$. Then there exists a mapping $ Col_f: \mathbb{Z}_q^k \to [\lambda] $ such that $ Col_f(\mathbf{u}) \neq Col_f(\mathbf{v})$ for any $ \mathbf{u}, \mathbf{v} \in \mathbb{Z}_q^k $ with $f(\mathbf{u}) \neq f(\mathbf{v})$ and $d_p(\mathbf{u}, \mathbf{v}) \le \rho$.
\end{lem} 
Next, we present an explicit construction of FCSPLCs for any pair-locally $(2t, \lambda)$-bounded function satisfying the contiguous block condition of Lemma~\ref{Cont_LB}.
\begin{cons}
	\label{Cons3} 
	Let $f\!\!:\mathbb{Z}_{q}^k\to \!\! Im(f) = \{f_1, f_2, \ldots, f_\lambda, \ldots, f_\mu\}$ be a pair-locally $(2t, \lambda)$-bounded function that satisfies the contiguous block condition of Lemma~\ref{Cont_LB} such that $\lambda \le \frac{q}{2}$ and $t$ denotes the number of symbol-pair errors to be corrected, where 
	$\mu \ge \lambda$ and $t \in \mathbb{N} $.  Define an encoding for $\mathbf{u}$ as $\mathrm{Enc}(\boldsymbol{u}) = (\boldsymbol{u}, \{p_{Col_f(\boldsymbol{u})}\}^r)$, where 
	\begin{equation*}
		p_{Col_f(\boldsymbol{u})} =
		\left\{
		\begin{array}{ll}
			\displaystyle 0 & \text{if }  Col_f(\boldsymbol{u}) = 1, \\[6pt]
			\displaystyle 2 \left \lfloor \frac{q}{2\lambda} \right \rfloor & \text{if }  Col_f(\boldsymbol{u}) = 2, \\[6pt]
			\displaystyle . & \ldots, \\[6pt]
			\displaystyle 2(i-1)\left \lfloor \frac{q}{2\lambda} \right \rfloor & \text{if } Col_f(\boldsymbol{u}) = i, \\[6pt]
			\displaystyle . & \ldots, \\[6pt]
			\displaystyle 2(\lambda-1)\left \lfloor \frac{q}{2\lambda} \right \rfloor & \text{if } Col_f(\boldsymbol{u}) = \lambda,
		\end{array}
		\right.
	\end{equation*}
	$i \in [\lambda]$, $r = {\left \lceil \frac{t}{\left \lfloor \frac{q}{2\lambda} \right \rfloor} \right \rceil}$ and \(\{p_{Col_f(\boldsymbol{u})}\}^r\) means \( r \)-fold repetition of the parity symbol $p_{Col_f(\boldsymbol{u})}$.
\end{cons} The following corollary plays a key role in establishing that \textit{construction~\ref{Cons3}} yields an FCSPLC for any pair-locally $(2t, \lambda)$-bounded function satisfying the contiguous block condition of Lemma~\ref{Cont_LB}.
\begin{cor}
	\label{cor_delta}
	Let $\operatorname{Enc}(\mathbf{u})
	=
	(\mathbf{u}, \, \{p_\mathbf{u}\}^{r})
	\,\text{and}\
	\operatorname{Enc}(\mathbf{v})
	=
	(\mathbf{v}, \, \{p_\mathbf{v}\}^{r}),$
	where $\{p_\mathbf{u}\}^{r}$ and $\{p_\mathbf{v}\}^{r}$ denote the $r$-fold repetitions of the
	symbols $p_\mathbf{u},p_\mathbf{v}\in\mathbb{Z}_q$, respectively, and $\mathbf u, \mathbf v \in \mathbb{Z}_q^k$. Then $d_p(\operatorname{Enc}(\mathbf{u}),\operatorname{Enc}(\mathbf{v}))
	\ge
	d_p(\mathbf{u},\mathbf{v})
	+
	d_p(\{p_\mathbf{u}\}^{r},\{p_\mathbf{v}\}^{r}).$
\end{cor}

\begin{IEEEproof}
	By Lemma~\ref{lem_delta}, $\Delta
	=
	\max\{a,b\}
	+
	\max\{c,d\}
	-
	\max\{a,d\}
	-
	\max\{c,b\},$
	where $a=d_L(\mathbf{u}_k,\mathbf{v}_k),\,
	b=d_L(p_\mathbf{u},p_\mathbf{v}),\,
	c=d_L(p_\mathbf{u},p_\mathbf{v}), \text{ and }
	d=d_L(\mathbf{u}_1,\mathbf{v}_1).$
	Since $b=c$, $\Delta
	=
	\max\{a,b\}
	+
	\max\{b,d\}
	-
	\max\{a,d\}
	-
	b.$
	
	We consider the following cases.
	
	\smallskip
	
	\emph{Case 1}: $b\ge a$ and $b\ge d$. Then $\Delta
	=
	b+b-\max\{a,d\}-b
	=
	b-\max\{a,d\}
	\ge0.$
	
    \smallskip
	
	\emph{Case 2}: $b\ge a$ and $b<d$. Therefore,  $\Delta
	=
	b+d-d-b
	=
	0.$
	
	\smallskip
	
	\emph{Case 3}: $b<a$ and $b\ge d$. Then $\Delta
	=
	a+b-a-b
	=
	0.$
	
	\smallskip
	
	\emph{Case 4}: $b<a$ and $b<d$. Therefore, $\Delta
	=
	a+d-\max\{a,d\}-b
	=
	\min\{a,d\}-b
	\ge0.$
	
	\smallskip
	
	Hence, $\Delta\ge0$ in all cases. Therefore,
	\[
	\begin{aligned}
		d_p(\operatorname{Enc}(\mathbf{u}),\operatorname{Enc}(\mathbf{v}))
		&=
		d_p(\mathbf{u},\mathbf{v})
		+
		d_p(\{p_\mathbf{u}\}^{r},\{p_\mathbf{v}\}^{r})
		+
		\Delta\\
		&\ge
		d_p(\mathbf{u},\mathbf{v})
		+
		d_p(\{p_\mathbf{u}\}^{r},\{p_\mathbf{v}\}^{r}).
	\end{aligned}
	\]
\end{IEEEproof} 
We now show that Construction~\ref{Cons3} yields an FCSPLC for any pair-locally $(2t, \lambda)$-bounded function satisfying the contiguous block condition of Lemma~\ref{Cont_LB}. Let $\mathbf{u_i}, \mathbf{u_j} \in \mathbb{Z}_q^k$ be such that $Col_f(\boldsymbol{u}_i) = i$ and $Col_f(\boldsymbol{u}_j) = j$, $\forall$ $i, j \in[\lambda]$ and $i \neq j$ for any $\lambda \le \frac{q}{2}$. By Corollary~\ref{cor_delta}, we have $$d_{p}(\mathrm{Enc}(\boldsymbol{u_i}), \mathrm{Enc}(\boldsymbol{u_j})) \ge d_{p}(\boldsymbol{u_i}, \boldsymbol{u_j}) + d_{p}(\{p_{Col_f(\boldsymbol{u_i})}\}^r, \{p_{Col_f(\boldsymbol{u_j})}\}^r),$$ where $p_{Col_f(\boldsymbol{u_i})}$ and $p_{Col_f(\boldsymbol{u_j})}$ are the parity symbols associated with $Col_f(\boldsymbol{u}_i) = i$ and $Col_f(\boldsymbol{u}_j) = j$, respectively. We consider two cases, \textit{Case 1}: If $d_p(\mathbf{u_i},\mathbf{u_j})\geq 2t+1$, then  $d_{p}(\mathrm{Enc}(\boldsymbol{u_i}), \mathrm{Enc}(\boldsymbol{u_j})) \ge d_p(\mathbf{u_i},\mathbf{u_j}) \ge 2t+1$ and \textit{Case 2}: If  $d_p(\mathbf{u_i},\mathbf{u_j}) \le 2t$, then $f(\mathbf{u_j})\in B_p^f(\mathbf{u_i},2t)$, and  $f(\mathbf{u_i}) \neq f(\mathbf{u_j}) $. Therefore, $d_{p}(p_{Col_f(\boldsymbol{u_i})}, p_{Col_f(\boldsymbol{u_j})}) \ge 2\left \lfloor \frac{q}{2\lambda} \right \rfloor$ from \textit{Construction~\ref{Cons3}} and $d_{p}(\{p_{Col_f(\boldsymbol{u_i})}\}^r, \{p_{Col_f(\boldsymbol{u_j})}\}^r) \ge 2 \left \lfloor \frac{q}{2\lambda} \right \rfloor r \ge  2 \left \lfloor \frac{q}{2\lambda} \right \rfloor \left \lceil \frac{t}{\left \lfloor \frac{q}{2\lambda} \right \rfloor} \right \rceil \ge 2t$. Since $\mathbf{u_i} \neq \mathbf{u_j}$, at least one coordinate differs. Consequently, at least two adjacent symbol pairs are affected, and hence $d_p(\mathbf{u_i},\mathbf{u_j})\geq 2$. Therefore,   \[
\begin{aligned}
	d_{p}\bigl(\mathrm{Enc}(\boldsymbol{u_i}),\mathrm{Enc}(\boldsymbol{u_j})\bigr)
	&\ge
	d_{p}(\boldsymbol{u_i},\boldsymbol{u_j})
	+
	d_{p}\bigl(\{p_{Col_f(\boldsymbol{u_i})}\}^{r},
	\{p_{Col_f(\boldsymbol{u_j})}\}^{r}\bigr)\\
	&\ge
	2+2t\\
	& >
	2t+1.
\end{aligned}
\]

Pair-locally binary functions are a special case of pair-locally $(\rho,\lambda)$-bounded functions with $\lambda=2$. Hence, the construction of FCSPLCs for pair-locally binary functions with $\rho=2t$ follows directly from Construction~\ref{Cons3}. Specifically, the parity symbols corresponding to the two color classes are $p_{Col_f(\boldsymbol{u_1})}=0
\,\text{and}\,
p_{Col_f(\boldsymbol{u_2})}=\frac{q}{2},$ respectively, 
and the number of parity symbols $r=\left\lceil\frac{2t-1}{\frac q2}\right\rceil.$
Let $\mathbf{u}_1,\mathbf{u}_2\in\mathbb{Z}_q^k$ satisfy
$f(\mathbf{u}_1)\neq f(\mathbf{u}_2)$.
As in the proof of Construction~\ref{Cons3}, we consider two cases, \emph{Case 1:} If
$d_p(\mathbf{u}_1,\mathbf{u}_2)\ge2t+1$,
then $d_p\bigl(\mathrm{Enc}(\mathbf{u}_1),\mathrm{Enc}(\mathbf{u}_2)\bigr)
\ge
d_p(\mathbf{u}_1,\mathbf{u}_2)
\ge
2t+1$ and \emph{Case 2:} if
$d_p(\mathbf{u}_1,\mathbf{u}_2)\le2t$.
Since
$f(\mathbf{u}_2)\in B_p^f(\mathbf{u}_1,2t)$
and
$f(\mathbf{u}_1)\neq f(\mathbf{u}_2)$,
\textit{Construction~\ref{Cons3}} gives $d_p\!\left(
p_{Col_f(\mathbf{u}_1)},
p_{Col_f(\mathbf{u}_2)}
\right)
=
\frac{q}{2}.$
Consequently, $d_p\!\left(
\{p_{Col_f(\mathbf{u}_1)}\}^{r},
\{p_{Col_f(\mathbf{u}_2)}\}^{r}
\right)
\ge
\frac{q}{2}r
\ge
\frac{q}{2}
\left\lceil
\frac{2t-1}{\frac q2}
\right\rceil
\ge
2t-1.$
Since $\mathbf{u}_1\neq\mathbf{u}_2$, we have
$d_p(\mathbf{u}_1,\mathbf{u}_2)\ge2$.
Therefore,
\[
\begin{aligned}
	d_p\bigl(\mathrm{Enc}(\mathbf{u}_1),\mathrm{Enc}(\mathbf{u}_2)\bigr)
	&\ge
	d_p(\mathbf{u}_1,\mathbf{u}_2)
	+
	d_p\!\left(
	\{p_{Col_f(\mathbf{u}_1)}\}^{r},
	\{p_{Col_f(\mathbf{u}_2)}\}^{r}
	\right)\\
	&\ge
	2+(2t-1)\\
	&=
	2t+1.
\end{aligned}
\]
Thus, \textit{Construction~\ref{Cons3}} yields an FCSPLC for the pair-locally binary function with an upper bound of $\left\lceil\frac{2t-1}{\frac q2}\right\rceil$ on the optimal redundancy, whereas for general pair-locally $(2t, \lambda)$-bounded functions, the corresponding upper bound is $\left \lceil \frac{t}{\left \lfloor \frac{q}{2\lambda} \right \rfloor} \right \rceil$. 

The following example gives an illustration of \textit{construction~\ref{Cons3}} for a locally $(2t, \lambda)_L$-bounded function.
\begin{example}
	Let \( f: \mathbb{Z}_8^k \to \mathrm{Im}(f) \) be a locally \((2,4)\)-bounded function satisfying the contiguous block condition in Lemma~\ref{Cont_LB} with parameters \(t=1\) and \(\lambda=4\), where \(|\mathrm{Im}(f)|\ge4\). Consider vectors \(\boldsymbol{u}_1,\boldsymbol{u}_2,\boldsymbol{u}_3,\boldsymbol{u}_4\in\mathbb{Z}_8^k\) such that
	$f(\boldsymbol{u}_i)\neq f(\boldsymbol{u}_j)
	\,\text{and}\,
	d_p(\boldsymbol{u}_i,\boldsymbol{u}_j)\le2,$
	for all distinct \(i,j\in[4]\). The parity symbols assigned by Construction~\ref{Cons3} are $p_1=0,\,
	p_2=2,\,
	p_3=4,
	\text{ and } p_4=6$
	corresponding to \(\mathrm{Col}_f(\boldsymbol{u})=1,2,3,\) and \(4\), respectively. Since $\lambda=\frac q2$,   Construction~\ref{Cons3} guarantees that $d_p(\mathrm{Enc}(\boldsymbol{u}_i),\mathrm{Enc}(\boldsymbol{u}_j))
	\ge
	d_p(\boldsymbol{u}_i,\boldsymbol{u}_j) \\
	+
	d_p\!\left(p_{\mathrm{Col}_f(\boldsymbol{u}_i)},
	p_{\mathrm{Col}_f(\boldsymbol{u}_j)}\right)
	\ge 2t+1=3,$
	for all \(i\neq j\in[4]\) as \(d_p(\boldsymbol{u}_i,\boldsymbol{u}_j)\ge2\) at least.
\end{example}
    \subsection{Pair-Lee Weight Function}
    The pair-Lee weight of a vector quantifies its total deviation from the zero vector and serves as a natural measure of symbol-pair error magnitude.
    \begin{defn}
    	The pair-Lee weight function is defined as $f(\mathbf{u})=w_p(\mathbf{u})
    	=
    	\sum_{i=1}^{k}
    	\max\{\mathrm{w}_L(u_i),\mathrm{w}_L(u_{i+1})\},$
    	where $\mathbf{u} \in\mathbb{Z}_q^k$,
    	$k\ge 2$, and $u_{k+1}=u_1$.
    \end{defn}
    The image of the pair-Lee weight function is given by $\mathrm{Im}(f)=
    \begin{cases}
    	\{0,2,4,\ldots,q\}, & k=2,\\[1ex]
    	\{0\}\cup\{2,3,\ldots,\tfrac{kq}{2}\}, & k\ge 3.
    \end{cases}$
    Consequently, the expressiveness of the pair-Lee weight function is $E=|\mathrm{Im}(f)|
    =
    \begin{cases}
    	\dfrac{q}{2}+1, & k=2,\\[1ex]
    	\dfrac{kq}{2}, & k\ge 3.
    \end{cases}$ \\ 
    The subsequent construction relies on the fact that the pair-Lee distance between two vectors is bounded below by the absolute difference of their pair-Lee weights. Indeed, since $w_p(\mathbf{u})=d_p(\mathbf{u},\mathbf{0}),$
    the triangle inequality implies $d_p(\mathbf{u},\mathbf{0})
    \le
    d_p(\mathbf{u},\mathbf{v})
    +
    d_p(\mathbf{v},\mathbf{0}),$
    or equivalently, $w_p(\mathbf{u})
    \le
    d_p(\mathbf{u},\mathbf{v})
    +
    w_p(\mathbf{v}).$
    Hence, $w_p(\mathbf{u})-w_p(\mathbf{v})
    \le
    d_p(\mathbf{u},\mathbf{v}).$
    Interchanging the roles of $\mathbf{u}$ and $\mathbf{v}$ gives $w_p(\mathbf{v})-w_p(\mathbf{u})
    \le
    d_p(\mathbf{u},\mathbf{v})$,
    and therefore, $d_p(\mathbf{u},\mathbf{v})
    \ge
    \left|w_p(\mathbf{u})-w_p(\mathbf{v})\right|$. We next present an explicit construction of FCSPLCs for the pair-Lee weight function.
    \begin{cons} \label{cons1}
    	For $\mathbf{u}\in\mathbb{Z}_q^k$ and $f(\mathbf{u})=w_p(\mathbf{u})$. Define $\mathrm{Enc}(\mathbf{u})
    	=
    	(\mathbf{u},\mathbf{p}_{w_p(\mathbf{u})+1})$,
    	where the parity vectors are chosen as follows. Let
    	$\mathbf{p}_1,\mathbf{p}_2,\ldots,\mathbf{p}_{2t+1}$
    	be the codewords of a symbol-pair Lee-distance code with minimum pair-Lee distance
    	$2t+\frac{q}{2}-1$, i.e., $d_p(\mathbf{p}_i,\mathbf{p}_j)
    	\ge
    	2t+\frac{q}{2}-1,
    	\,
    	0\le i<j\le 2t+1$.
    	Set $\mathbf{p}_i
    	=
    	\mathbf p_{i \operatorname{smod} (2t+1)}$ for each $i\ge 2t+2, \, \text{ where } a \operatorname{smod} b=\big((a-1) \bmod b\big)+1$.
    \end{cons}
    We now show that \textit{Construction~\ref{cons1}} defines an FCSPLC for the pair-Lee weight function. Let
    $\mathbf{u}_i,\mathbf{u}_j\in\mathbb{Z}_q^k$ with $w_p(\mathbf{u}_i)\neq w_p(\mathbf{u}_j)$. If $|w_p(\mathbf{u}_i)-w_p(\mathbf{u}_j)| \ge 2t+1,$ then
    \[
    \begin{aligned}
    	d_p(\mathrm{Enc}(\mathbf{u}_i),\mathrm{Enc}(\mathbf{u}_j))
    	&\ge
    	d_p(\mathbf{u}_i,\mathbf{u}_j) \!\!
    	&\ge
    	|w_p(\mathbf{u}_i)-w_p(\mathbf{u}_j)| 
    	&\ge
    	2t+1.
    \end{aligned} 
    \] if $|w_p(\mathbf{u}_i)-w_p(\mathbf{u}_j)|
    \le
    2t$, then by Lemma~\ref{lem_delta},
    \[
    \begin{aligned}
    	d_p(\mathrm{Enc}(\mathbf{u}_i),\mathrm{Enc}(\mathbf{u}_j))
    	&\ge
    	d_p(\mathbf{u}_i,\mathbf{u}_j)
    	+
    	d_p(\mathbf{p}_{w_p(\mathbf{u}_i)},
    	\mathbf{p}_{w_p(\mathbf{u}_j)})
    	-\frac{q}{2}\\
    	&\ge
    	2+
    	\left(2t+\frac{q}{2}-1\right)
    	-\frac{q}{2}\\
    	&=
    	2t+1.
    \end{aligned}
    \]
    Hence, $d_p(\mathrm{Enc}(\mathbf{u}_i),\mathrm{Enc}(\mathbf{u}_j))
    \ge
    2t+1$
    for every pair of information vectors having distinct pair-Lee weights. Therefore, the proposed encoding defines an FCSPLC for the pair-Lee weight function. Moreover, the above construction yields $r_p^{w_p}(q,k,t)
    \le
    N_p\!\left(2t+1,\,
    2t+\frac{q}{2}-1\right)$, an upper bound on the optimal redundancy, 
    where $N_p\!\left(2t+1,\,
    2t+\frac{q}{2}-1\right)$ denotes the shortest length of a symbol-pair Lee-distance code with $2t+1$ codewords and minimum pair-Lee distance $2t+\frac{q}{2}-1$. 
    
    Using \textit{Construction~\ref{cons1}}, we obtain the following upper bounds on the optimal redundancy of FCSPLCs over $\mathbb{Z}_4$.
    \begin{cor}
    	For $q=4$ and $k\ge2$, we have $r_p^{w_p}(4,k,1)\le2,\,
    	r_p^{w_p}(4,k,2)\le3,\,
    	r_p^{w_p}(4,k,3)\le5.$
    \end{cor}
    \begin{IEEEproof}
    	For each value of $t$, we find $2t+1$ parity vectors whose minimum pair-Lee distance is
    	$2t+\frac{q}{2}-1$. Applying \textit{Construction~\ref{cons1}} then immediately yields the corresponding upper bound on the optimal redundancy.
    	
    	For $t=1$, let $\mathbf{p}_1=00,\,
    	\mathbf{p}_2=02,\,
    	\mathbf{p}_3=20,$ and $\mathbf{p}_i
    	=\mathbf{p}_{\,i \text{ smod } 3}$ for $i\ge 4$.
    	Since these parity vectors have minimum pair-Lee distance $3$, $r_p^{w_p}(4,k,1)\le2.$
    	
    	For $t=2$, let $\mathbf{p}_1=000,\;
    	\mathbf{p}_2=012,\;
    	\mathbf{p}_3=201,\;
    	\mathbf{p}_4=222,\;
    	\mathbf{p}_5=120,$ and $\mathbf{p}_i
    	=\mathbf{p}_{\,i \text{ smod } 5}$ for $i\ge 6$.
    	Since these parity vectors have minimum pair-Lee distance $5$, $r_p^{w_p}(4,k,2)\le3.$
    	
    	For $t=3$, let $\mathbf{p}_1=00000,\;
    	\mathbf{p}_2=00122,\;
    	\mathbf{p}_3=01201,\;
    	\mathbf{p}_4=01323,\;
    	\mathbf{p}_5=02002,\;
    	\mathbf{p}_6=02120,\;
    	\mathbf{p}_7=03203,$ and $\mathbf{p}_i
    	=\mathbf{p}_{\,i \text{ smod } 7}$ for $i\ge 8$.
    	Since these parity vectors have minimum pair-Lee distance $7$, $r_p^{w_p}(4,k,3)\le5.$   	
    \end{IEEEproof}
    \subsection{Pair-Lee Weight Distribution Function}
    We next introduce the pair-Lee weight distribution function, an important function that helps characterize the pair-weight distribution of codewords under the Lee metric.
    \begin{defn}
    	The pair-Lee weight distribution function is defined as $f(\mathbf u)=\Delta_p^T(\mathbf u)\triangleq \left\lfloor \frac{\mathrm{w_p}(\boldsymbol{u})}{T} \right\rfloor$, where $\boldsymbol{u} \in \mathbb{Z}_q^k $ and $T, k \in \mathbb{N}$ with $k \ge 2$. 
    \end{defn}
    The expressiveness of the pair-Lee weight distribution function is $\mathrm{E}=\left \lfloor \frac {kq} {2T} \right \rfloor+1$. The function is a step threshold function based on the pair-Lee weight of $\mathbf{u}$ with $\mathrm E-1$ steps. The function value increases by one at every multiple of $T$. For FCSPLCs, we only consider the case where $ \ge 2t+1$. We now present an explicit construction of FCSPLCs for the pair-Lee weight distribution function. The construction assigns the same parity vector to all information vectors with the same pair-Lee weight, whereas information vectors with different pair-Lee weights are assigned different parity vectors as shown in Table~\ref{tab:pairweight}.
    \begin{table}[!htbp]
    	\caption{Parity vectors for the pair-Lee weight distribution function}
    	\label{tab:pairweight}
    	\centering
    	\renewcommand{\arraystretch}{1.15}
    	\begin{tabular}{ccc}
    		\toprule
    	$\mathrm{w}_p(\mathbf{u})$ &
    	$\mathbf{p}_{(\mathrm{w}_p(\mathbf{u})+1)\,\operatorname{smod}\,T}$ &
    	$\Delta_p^T(\mathbf u)$  \\
    		\midrule
    		$mT-T$     & $\mathbf{p}_1$ & \multirow{5}{*}{$m-1$} \\
    		$mT-T+1$   & $\mathbf{p}_2$ & \\
    		$\vdots$   & $\vdots$         & \\
    		$mT-2$     & $\mathbf{p}_{T-1}$ & \\
    		$mT-1$     & $\mathbf{p}_T$ & \\
    		\midrule
    		$mT$       & $\mathbf{p}_1$ & \multirow{3}{*}{$m$} \\
    		$mT+1$     & $\mathbf{p}_2$ & \\
    		$\vdots$   & $\vdots$         & \\
    		\bottomrule
    	\end{tabular}
    \end{table}
    \begin{cons}\label{cons2}
    	Let $f(\mathbf u)=\Delta_p^T(\mathbf u)= \left\lfloor \frac{\mathrm{w_p}(\boldsymbol{u})}{T} \right\rfloor$ for $\mathbf{u} \in \mathbb{Z}_q^k$, where $q=2^m, m \ge 2$ and $T \ge 2t+1$. Define $\operatorname{Enc}(\mathbf{u})
    	=
    	\left(\mathbf{u},\mathbf{p}_{(w_p(\mathbf{u})+1) \operatorname{smod} T}\right),$ where the parity vectors
    	$\mathbf{p}_1,\mathbf{p}_2,\ldots,\mathbf{p}_T\in\mathbb{Z}_q^{t+1}$ are 
    \[
    \mathbf{p}_i=
    \begin{cases}
    	\{\frac{q}{2}\}^j\{0\}^{\,t+1-j},
    	& i=2j,\;1\le j\le t,\\[1ex]
    	0\{\frac{q}{2}\}^j\{0\}^{\,t-j},
    	& i=2j+1,\;0\le j\le t-1,\\[1ex]
    	\{\frac{q}{2}\}^{t+1},
    	& 2t+1\le i\le T.
    \end{cases}
    \]
    \end{cons}
    We show that \textit{Construction~\ref{cons2}} gives an FCSPLC for the pair-Lee weight distribution function. Let $\mathbf{u}_1,\mathbf{u}_2\in\mathbb{Z}_q^k$ be such that
    $f(\mathbf{u}_1)\neq f(\mathbf{u}_2)$. If
    $d_p(\mathbf{u}_1,\mathbf{u}_2)\ge2t+1$, then $d_p\bigl(\operatorname{Enc}(\mathbf{u}_1),\operatorname{Enc}(\mathbf{u}_2)\bigr)
    \ge2t+1.$
    Hence, it suffices to consider the case
    $d_p(\mathbf{u}_1,\mathbf{u}_2)\le2t$. Without loss of generality, assume that
    $w_p(\mathbf{u}_1)<w_p(\mathbf{u}_2)$. Since
    $d_p(\mathbf{u}_1,\mathbf{u}_2)\le2t<T$, we write $w_p(\mathbf{u}_1)=(l-1)T+w_1,\,
    w_p(\mathbf{u}_2)=lT+w_2,$
    where $l\ge1$. If $l=1$, then
    $w_1\in\{0,2,\ldots,T-1\}$ and
    $w_2\in\{0,1,\ldots,T-1\}$, and if $l\ge2$, then
    $w_1,w_2\in\{0,1,\ldots,T-1\}$. Furthermore, $w_p(\mathbf{u}_2)-w_p(\mathbf{u}_1)
    \le
    d_p(\mathbf{u}_1,\mathbf{u}_2)
    \le2t<T,$
    which implies $T+w_2-w_1< T,$
    and hence $w_1>w_2.$ 
    Therefore, by \textit{Construction~\ref{cons2}},
    \[
    \operatorname{Enc}(\mathbf{u}_1)
    =
    (\mathbf{u}_1,\mathbf{p}_{w_1+1}),
    \quad
    \operatorname{Enc}(\mathbf{u}_2)
    =
    (\mathbf{u}_2,\mathbf{p}_{w_2+1}).
    \]
  If $T=2t+1$, then,
 by \textit{Construction~\ref{cons2}}, the parity vectors
  $\mathbf{p}_{w_1+1}$ and $\mathbf{p}_{w_2+1}$ fall into one of the following cases.
  
  \smallskip
  
  \noindent\textit{Case 1}:
  $\mathbf{p}_{w_1+1}
  =\{\frac{q}{2}\}^{\frac{w_1+1}{2}}
  \{0\}^{\,t+1-\frac{w_1+1}{2}}$, $1\le w_1 \le 2t-1$ and
  $\, \mathbf{p}_{w_2+1}
  =\{\frac{q}{2}\}^{\frac{w_2+1}{2}}
  \{0\}^{\,t+1-\frac{w_2+1}{2}},$  $1\le w_2 \le 2t-1$.
  
  \smallskip
  
  \noindent\textit{Case 2}:
  $\mathbf{p}_{w_1+1}
  =0\{\frac{q}{2}\}^{\frac{w_1}{2}}
  \{0\}^{\,t-\frac{w_1}{2}},$ $0 \le w_1 \le 2t-2$ and
  $\, \mathbf{p}_{w_2+1}
  =0\{\frac{q}{2}\}^{\frac{w_2}{2}}
  \{0\}^{\,t-\frac{w_2}{2}},$ $0 \le w_2 \le 2t-2.$
  
  \smallskip
  
  \noindent\textit{Case 3}:
  $\mathbf{p}_{w_1+1}
  =\{\frac{q}{2}\}^{\frac{w_1+1}{2}}
  \{0\}^{\,t+1-\frac{w_1+1}{2}},$ $1\le w_1 \le 2t-1$ and
  $\, \mathbf{p}_{w_2+1}
  =0\{\frac{q}{2}\}^{\frac{w_2}{2}}
  \{0\}^{\,t-\frac{w_2}{2}},$ $0 \le w_2 \le 2t-2.$
  
  \smallskip
  
  \noindent\textit{Case 4}:
  $\mathbf{p}_{w_1+1}
  =0\{\frac{q}{2}\}^{\frac{w_1}{2}}
  \{0\}^{\,t-\frac{w_1}{2}},$ $0 \le w_1 \le 2t-2$ and
  $\, \mathbf{p}_{w_2+1}
  =\{\frac{q}{2}\}^{\frac{w_2+1}{2}}
  \{0\}^{\,t+1-\frac{w_2+1}{2}},$ $1\le w_2 \le 2t-1$.
  
  \smallskip
  
   \noindent\textit{Case 5}:
  $\mathbf{p}_{w_1+1}
  =\{\frac{q}{2}\}^{\frac {w_1+2} {2}}$, $w_1=2t$ and
  $\, \mathbf{p}_{w_2+1}
  =\{\frac{q}{2}\}^{\frac{w_2+1}{2}}
  \{0\}^{\,t+1-\frac{w_2+1}{2}},$  $1\le w_2 \le 2t-1$.
  
  \smallskip
  
  \noindent\textit{Case 6}:
  $\mathbf{p}_{w_1+1}
  =\{\frac{q}{2}\}^{\frac {w_1+2} {2}}$, $w_1=2t$ and
  $\, \mathbf{p}_{w_2+1}
  =0\{\frac{q}{2}\}^{\frac{w_2}{2}}
  \{0\}^{\,t-\frac{w_2}{2}},$ $0 < w_2 \le 2t-2.$ 
  
  \smallskip
  
  For \textit{Cases~1} and~\textit{2}, $d_p(\mathbf{p}_{w_1+1},\mathbf{p}_{w_2+1})
  =
  \frac{q}{4}(w_1-w_2+2),$
  whereas for \textit{Cases~3} and~\textit{4}, $d_p(\mathbf{p}_{w_1+1},\mathbf{p}_{w_2+1})
  \ge
  \frac{q}{4}(w_1-w_2+2),$
  and for \textit{Cases~5} and~\textit{6}, $d_p(\mathbf{p}_{w_1+1},\mathbf{p}_{w_2+1})
  =
  \frac{q}{4}(w_1-w_2+3), \text{ and } d_p(\mathbf{p}_{w_1+1},\mathbf{p}_{w_2+1})
  =
  \frac{q}{4}(w_1-w_2+4) \text{, respectively. }$
  
  For $w_1=2t$ and $w_2=0$, then, by \textit{Construction~\ref{cons2}}, $\mathbf{p}_{w_1+1}
  =
  \mathbf{p}_{2t+1}
  =
  \{\frac{q}{2}\}^{t+1},
  \,
  \mathbf{p}_{w_2+1}
  =
  \mathbf{p}_1
  =
  \{0\}^{t+1}.$
By Corollary~\ref{cor_delta2}, $d_p(\operatorname{Enc}(\mathbf{u}_1),\operatorname{Enc}(\mathbf{u}_2))
  \ge
  d_p(\mathbf{u}_1,\mathbf{u}_2)
  +
  d_p(\mathbf{p}_{2t+1},\mathbf{p}_1).$
  Hence,
  \[
  \begin{aligned}
  	d_p(\operatorname{Enc}(\mathbf{u}_1),\operatorname{Enc}(\mathbf{u}_2)) 
  	& \ge
  	2+d_p(\mathbf{p}_{2t+1},\mathbf{p}_1) (\text{ since } d_p(\mathbf{u}_1,\mathbf{u}_2)=2) \\
  	& \ge
  	2+\frac{q}{2}(t+1)\\
  	&>
  	2t+1 \, (\text{since } q\ge4).
  \end{aligned}
  \]
  
  For all other cases, $d_p(\mathbf{p}_{w_1+1},\mathbf{p}_{w_2+1})
  \ge
  \frac{q}{4}(w_1-w_2+2).$ Hence,
  \[
  \begin{aligned}
  	d_p(\operatorname{Enc}(\mathbf{u}_1),\operatorname{Enc}(\mathbf{u}_2))
  	&\ge
  	\frac{q}{4}(w_1-w_2+2)
  	+(T+w_2-w_1)
  	-\frac{q}{2}\\
  	&=
  	T+\left(\frac{q}{4}-1\right)(w_1-w_2)\\
  	&=
  	2t+1+\left(\frac{q}{4}-1\right)(w_1-w_2)\\
  	&>
  	2t+1,
  \end{aligned}
  \]
  where the final inequality follows from $q\ge4$ and $w_1>w_2$.
  
 If $T>2t+1$, note that $w_p(\mathbf{u}_2)-w_p(\mathbf{u}_1)
 =
 T+w_2-w_1
 \le
 d_p(\mathbf{u}_1,\mathbf{u}_2)
 \le
 2t.$
 Hence, $w_1-w_2
 \ge
 T-2t.$
 Since $T>2t+1$, we have $T-2t\ge2,$
 which implies $w_1-w_2\ge2.$ We first show that the case $\mathbf{p}_{w_1+1}
 =
 \mathbf{p}_{w_2+1}
 =
 \left\{\frac{q}{2}\right\}^{t+1}$
 cannot occur. Indeed, this would require $2t\le w_2\le w_1\le T-1.$
 Since $w_p(\mathbf{u}_2)-w_p(\mathbf{u}_1)
 =
 T+w_2-w_1,$
 we obtain $w_p(\mathbf{u}_2)-w_p(\mathbf{u}_1)
 \ge
 T+2t-(T-1)
 =
 2t+1.$
 Hence, $d_p(\mathbf{u}_1,\mathbf{u}_2)
 \ge
 w_p(\mathbf{u}_2)-w_p(\mathbf{u}_1)
 \ge
 2t+1,$
 which contradicts the assumption
 $d_p(\mathbf{u}_1,\mathbf{u}_2)\le2t$.
 Therefore, at most one of the parity vectors can be equal to
 $\{\frac{q}{2}\}^{t+1}$.
 
 Consequently, by \textit{Construction~\ref{cons2}}, the parity vectors
 $\mathbf{p}_{w_1+1}$ and $\mathbf{p}_{w_2+1}$ fall into one of the following cases. \textit{Cases~a--d} are identical to \textit{Cases~1--4}, respectively.
 
 \smallskip
 
 \noindent\textit{Case a}:
 $\mathbf{p}_{w_1+1}
 =
 \{\frac{q}{2}\}^{\frac{w_1+1}{2}}
 \{0\}^{\,t+1-\frac{w_1+1}{2}},$
 $1\le w_1\le2t-1$
 and
 $\mathbf{p}_{w_2+1}
 =
 \{\frac{q}{2}\}^{\frac{w_2+1}{2}}
 \{0\}^{\,t+1-\frac{w_2+1}{2}},$
 $1\le w_2\le2t-1.$
 
 \smallskip
 
 \noindent\textit{Case b}:
 $\mathbf{p}_{w_1+1}
 =
 0\{\frac{q}{2}\}^{\frac{w_1}{2}}
 \{0\}^{\,t-\frac{w_1}{2}},$
 $0\le w_1\le2t-2$
 and
 $\mathbf{p}_{w_2+1}
 =
 0\{\frac{q}{2}\}^{\frac{w_2}{2}}
 \{0\}^{\,t-\frac{w_2}{2}},$
 $0\le w_2\le2t-2.$
 
 \smallskip
 
 \noindent\textit{Case c}:
 $\mathbf{p}_{w_1+1}
 =
 \{\frac{q}{2}\}^{\frac{w_1+1}{2}}
 \{0\}^{\,t+1-\frac{w_1+1}{2}},$
 $1\le w_1\le2t-1$
 and
 $\mathbf{p}_{w_2+1}
 =
 0\{\frac{q}{2}\}^{\frac{w_2}{2}}
 \{0\}^{\,t-\frac{w_2}{2}},$
 $0\le w_2\le2t-2.$
 
 \smallskip
 
 \noindent\textit{Case d}:
 $\mathbf{p}_{w_1+1}
 =
 0\{\frac{q}{2}\}^{\frac{w_1}{2}}
 \{0\}^{\,t-\frac{w_1}{2}},$
 $0\le w_1\le2t-2$
 and
 $\mathbf{p}_{w_2+1}
 =
 \{\frac{q}{2}\}^{\frac{w_2+1}{2}}
 \{0\}^{\,t+1-\frac{w_2+1}{2}},$
 $1\le w_2\le2t-1.$
 
 \smallskip
 
 \noindent\textit{Case e}:
 $\mathbf{p}_{w_1+1}
 =
 \{\frac{q}{2}\}^{t+1},$
 $2t\le w_1\le T-1$
 and
 $\mathbf{p}_{w_2+1}
 =
 \{\frac{q}{2}\}^{\frac{w_2+1}{2}}
 \{0\}^{\,t+1-\frac{w_2+1}{2}},$
 $1\le w_2\le2t-1.$
 
 \smallskip
 
 \noindent\textit{Case f}:
 $\mathbf{p}_{w_1+1}
 =
 \{\frac{q}{2}\}^{t+1},$
 $2t\le w_1\le T-1$
 and
 $\mathbf{p}_{w_2+1}
 =
 0\{\frac{q}{2}\}^{\frac{w_2}{2}}
 \{0\}^{\,t-\frac{w_2}{2}},$
 $0 < w_2\le2t-2.$ 
 
 \smallskip
 
 For \textit{Cases~a} and~\textit{b}, $d_p(\mathbf{p}_{w_1+1},\mathbf{p}_{w_2+1})
 =
 \frac{q}{4}(w_1-w_2+2),$
 whereas for \textit{Cases~c} and~\textit{d}, $d_p(\mathbf{p}_{w_1+1},\mathbf{p}_{w_2+1})
 \ge
 \frac{q}{4}(w_1-w_2+2).$
 For \textit{Cases~e} and~\textit{f}, $d_p(\mathbf{p}_{w_1+1},\mathbf{p}_{w_2+1})
 =
 \frac{q}{4}(2t-w_2+3),$
 and $d_p(\mathbf{p}_{w_1+1},\mathbf{p}_{w_2+1})
 =
 \frac{q}{4}(2t-w_2+4),$
 respectively.
 
 For $w_1=T-1$ and $w_2=0$, $\mathbf{p}_{w_1+1}
 =
 \mathbf{p}_T
 =
 \left\{\frac{q}{2}\right\}^{t+1},
 \,
 \mathbf{p}_{w_2+1}
 =
 \mathbf{p}_1
 =
 \{0\}^{t+1}.$
 By Corollary~\ref{cor_delta2}, $d_p(\operatorname{Enc}(\mathbf{u}_1),\operatorname{Enc}(\mathbf{u}_2))
 \ge
 d_p(\mathbf{u}_1,\mathbf{u}_2)
 +
 d_p(\mathbf{p}_T,\mathbf{p}_1),$
 and we have $d_p(\mathbf{u}_1,\mathbf{u}_2)=2.$
 Hence,
 \[
 \begin{aligned}
 	d_p(\operatorname{Enc}(\mathbf{u}_1),\operatorname{Enc}(\mathbf{u}_2))
 	&\ge
 	2+\frac{q}{2}(t+1)\\
 	&>
 	2t+1 \, (\text{since } q\ge4).
 \end{aligned}
 \] 
 
 For \textit{Cases~a--d}, $d_p(\mathbf{p}_{w_1+1},\mathbf{p}_{w_2+1})
 \ge
 \frac{q}{4}(w_1-w_2+2).$
 Hence,
 \[
 \begin{aligned}
 	d_p(\operatorname{Enc}(\mathbf{u}_1),\operatorname{Enc}(\mathbf{u}_2))
 	&\ge
 	\frac{q}{4}(w_1-w_2+2)
 	+(T+w_2-w_1)
 	-\frac{q}{2}\\
 	&=
 	T+\left(\frac{q}{4}-1\right)(w_1-w_2)\\
 	&\ge
 	T>2t+1,
 \end{aligned}
 \]
 where the last inequality follows from $q\ge4$ and $T>2t+1$.
 
 For \textit{Case~e},
 \[
 \begin{aligned}
 	d_p(\operatorname{Enc}(\mathbf{u}_1),\operatorname{Enc}(\mathbf{u}_2))
 	&\ge
 	\frac{q}{4}(2t-w_2+3)
 	+(T+w_2-w_1)
 	-\frac{q}{2}\\
 	&=
 	(T-w_1)
 	+\frac{q}{4}(2t-w_2+1)
 	+w_2.
 \end{aligned}
 \]
 Since $2t\le w_1\le T-1$, we have $T-w_1\ge1.$
 Moreover, $\frac{q}{4}(2t-w_2+1)+w_2
 =
 \frac{q}{4}(2t+1)
 +\left(1-\frac{q}{4}\right)w_2.$
 Since $q\ge4$, the coefficient
 $1-\frac{q}{4}\le0$, and thus the above expression is minimized when
 $w_2=2t-1$. Consequently,
 \[
 \begin{aligned}
 	d_p(\operatorname{Enc}(\mathbf{u}_1),\operatorname{Enc}(\mathbf{u}_2))
 	&\ge
 	1+\frac{q}{2}+2t-1\\
 	&=
 	2t+\frac{q}{2}>2t+1.
 \end{aligned}
 \]
 
For \textit{Case~f},
 \[
 \begin{aligned}
 	d_p(\operatorname{Enc}(\mathbf{u}_1),\operatorname{Enc}(\mathbf{u}_2))
 	&\ge
 	\frac{q}{4}(2t-w_2+4)
 	+(T+w_2-w_1)
 	-\frac{q}{2}\\
 	&=
 	(T-w_1)
 	+\frac{q}{4}(2t-w_2+2)
 	+w_2.
 \end{aligned}
 \]
 Again, $T-w_1\ge1$. Furthermore, $\frac{q}{4}(2t-w_2+2)+w_2
 =
 \frac{q}{4}(2t+2)
 +\left(1-\frac{q}{4}\right)w_2,$
 which is minimized when
 $w_2=2t-2$. Hence,
 \[
 \begin{aligned}
 	d_p(\operatorname{Enc}(\mathbf{u}_1),\operatorname{Enc}(\mathbf{u}_2))
 	&\ge
 	1+q+2t-2\\
 	&=
 	2t+q-1>2t+1.
 \end{aligned}
 \]
 
 \textit{Construction~\ref{cons2}} establishes an upper bound of $t+1$ on the optimal redundancy of FCSPLCs for the pair-Lee weight distribution function. The following example illustrates Construction~\ref{cons2} for the pair-Lee weight distribution function.
 \begin{example}
 	Consider $f(\mathbf{u})=\Delta_p^T(\mathbf u)= \left\lfloor \frac{\mathrm{w_p}(\boldsymbol{u})}{T} \right\rfloor$, where $\mathbf{u} \in \mathbb{Z}_4^5$ and $T=5$. The expressiveness $E = \left \lfloor \frac {kq} {2T} \right \rfloor+1 = 3$, i.e., \( f(\boldsymbol{u}) \in \{0,1,2\} \). Thus, possible pair-Lee weight values are $\mathrm{w_p}(\boldsymbol{u})=\{0,1,2,3,4,5,6,7, \\ 8,9,10\}$.  By Construction~\ref{cons2}, the assigned parity vectors for $t=2$ are $\mathbf{p_1}=(000), \mathbf{p_2}=(200), \mathbf{p_3}=(020), \mathbf{p_4}=(220), \mathbf{p_5}=(222), \mathbf{p_6}=\mathbf{p_1}=(000), \mathbf{p_7}=\mathbf{p_2}=(200), \mathbf{p_8}=\mathbf{p_3}=(020), \mathbf{p_9}=\mathbf{p_4}=(220), \mathbf{p_{10}}=\mathbf{p_5}=(222), \mathbf{p_{11}}=\mathbf{p_1}=(000)$, respectively.
 \end{example}
\section{A Plotkin-type Bound on the Optimal Redundancy of FCSPLCs for Linear Functions} \label{section6}
In this section, we derive a Plotkin-type bound on the optimal redundancy of FCSPLCs for linear functions. Linear functions play an important role in the study of FCCs under the Lee metric and were previously investigated in \cite{GA}. 
\begin{defn}[Linear Function \cite{GA}]
	A function $f : \mathbb{Z}_{q}^{k} \to \mathbb{Z}_{q}^{l}$ is linear if
	$f(\alpha \mathbf{x}+\beta \mathbf{y})
	=
	\alpha f(\mathbf{x})+\beta f(\mathbf{y})$
	for all $\mathbf{x},\mathbf{y}\in\mathbb{Z}_{q}^{k}$ and
	$\alpha,\beta\in\mathbb{Z}_{q}$.
\end{defn} 
A linear function can be expressed as a matrix operation $f(\mathbf u)=\mathbf F \mathbf u,$ where $\mathbf F \in \mathbb{Z}_q^{l \times k}$. The kernel of $f$, equivalently the null space of $\mathbf F$, is denoted by $\ker(f)$. Furthermore, $|\ker(f)|_{d}$ denotes the number of vectors in $\ker(f)$ with pair-Lee weight $d$, and $\mathbb{Z}_{q}^{k}/\ker(f)$ denotes the set of cosets of $\ker(f)$. We consider only linear functions with $l \le k$, where $\mathbf F \in \mathbb{Z}_q^{l \times k}$ is of full rank, i.e., $\operatorname{rank}_{\mathbb Z_q}(\mathbf F)=l$. 

The following lemma establishes a Plotkin-type bound on the optimal redundancy of FCSPLCs for linear functions.
\begin{lem}
	For a linear function $f: \mathbb Z_q^k \to \mathbb Z_q^l,$ the optimal redundancy of FCSPLCs is lower bounded by 
	\begin{equation}
	 r_{p}^f(q, k, t) \ge \frac {8} {3q}\Big[(2t+1-\frac q2)(1-q^{-l})-\frac {k(2q^2+1)}{6q}+\frac {s}{q^k}\Big],
	\end{equation} where $s=\sum_{\mathbf u \in \ker(f)} \mathrm{w}_p(\mathbf u)$.
\end{lem}
\begin{IEEEproof}
	Consider the message set $q^k = \{\mathbf{u}_i : i \in [q^k] \}$. Let $\mathbf{D}_f^{(1)}(t, \mathbf{u}_1, \dots, \mathbf{u}_{q^k})$ be the DRM of $f: \mathbb Z_q^k \to \mathbb Z_q^l$. For convenience, we denote this DRM by $\mathbf{D}$ throughout this proof. Let $\mathcal{P} = \{ \mathbf{p}_1, \mathbf{p}_2, \ldots, \mathbf{p}_{q^k} \} \subseteq q^r$ be a $\mathbf{D}_p$-code corresponding to the DRM $\mathbf{D}$. From Definition \ref{D_p-code}, it follows that $d_{p}(\mathbf{p}_i, \mathbf{p}_j) \geq [\mathbf{D}]_{ij}, \forall \, i, j \in [q^k].$  Therefore, we have
	\begin{equation}\label{eq:sum} 
		\sum_{i,j \in [q^k]} d_p(\mathbf{p}_i, \mathbf{p}_j) \ge \sum_{i,j \in [q^k]} [\mathbf{D}]_{ij}.
	\end{equation}
    Since $\operatorname{rank}_{\mathbb{Z}_q}(\mathbf{F})=l$, it follows that $\dim_{\mathbb{Z}_q}(\ker(f))=k-l.$
	Hence, each coset of $\ker(f)$ contains $q^{k-l}$ elements. If
	$\mathbf{u}_i,\mathbf{u}_j\in\mathbb{Z}_q^k$ belong to the same coset of
	$\ker(f)$, then $f(\mathbf{u}_i)=f(\mathbf{u}_j)$, and consequently,
	$[\mathbf{D}]_{ij}=0$. Therefore, every column of $\mathbf{D}$ contains at
	least $q^{k-l}$ zero entries.
	
	Suppose $\mathbf{u}_i$ and $\mathbf{u}_j$ belong to different cosets of
	$\ker(f)$. Then $[\mathbf{D}]_{ij}
	=
	\left[\left(2t+1-\frac{q}{2}\right)-d_p(\mathbf{u}_i,\mathbf{u}_j)\right]^+.$
	Now consider any other column $j'\neq j$ of $\mathbf{D}$ and let $\mathbf{v}=\mathbf{u}_{j'}-\mathbf{u}_j.$
	For each row index $i$, let $i'$ be the unique index satisfying $\mathbf{u}_{i'}=\mathbf{u}_i+\mathbf{v}.$
	Since addition by a fixed vector is a bijection on $\mathbb{Z}_q^k$, the
	mapping $i\mapsto i'$ is a permutation of the row indices. Moreover, $d_p(\mathbf{u}_{i'},\mathbf{u}_{j'})
	=
	d_p(\mathbf{u}_i+\mathbf{v},\mathbf{u}_j+\mathbf{v})
	=
	d_p(\mathbf{u}_i,\mathbf{u}_j),$
	where the last equality follows from the translation invariance of the
	symbol-pair Lee distance.
	
	Since $f$ is linear, $f(\mathbf{u}_{i'})=f(\mathbf{u}_i)+f(\mathbf{v})
	\quad\text{and}\quad
	f(\mathbf{u}_{j'})=f(\mathbf{u}_j)+f(\mathbf{v}),$
	which implies $f(\mathbf{u}_{i'})=f(\mathbf{u}_{j'})
	\iff
	f(\mathbf{u}_i)=f(\mathbf{u}_j).$
	Hence, $[\mathbf{D}]_{i'j'}
	=
	[\mathbf{D}]_{ij}.$ Therefore, every column of $\mathbf{D}$ is a permutation of every other
	column. Consequently,
	\begin{align}
		\sum_{i,j}[\mathbf{D}]_{ij}
		&=
		(\text{number of columns})
		\times
		(\text{sum of one column}) \nonumber\\
		&=
		q^k\sum_i[\mathbf{D}]_{ij}, \nonumber
	\end{align} 	where $q^k$ is the number of columns.

    Without loss of generality, let $\mathbf{u}_j=\mathbf{0}$. Let $I=\mathbb{Z}_q^k\setminus\ker(f)$
    denote the set of vectors outside the kernel of $f$. Thus, we have \begin{align}
    	[\mathbf{D}]_{ij}
    	&=
    	\left[\left(2t+1-\frac{q}{2}\right)-d_p(\mathbf{u}_i,\mathbf{0})\right]^+ \nonumber\\
    	&\ge
    	\left(2t+1-\frac{q}{2}\right)-d_p(\mathbf{u}_i,\mathbf{0}) \nonumber\\
    	&=
    	2t+1-\frac{q}{2}-\mathrm{w}_p(\mathbf{u}_i),
    	\qquad \forall\,\mathbf{u}_i\in I. \nonumber
    \end{align}
    Therefore,
    \begin{align}
    	\sum_i[\mathbf{D}]_{ij}
    	&\ge
    	(q^k-q^{k-l})\left(2t+1-\frac{q}{2}\right)
    	-\sum_{\mathbf{u}_i\in I}\mathrm{w}_p(\mathbf{u}_i). \nonumber
    \end{align}The sum of the pair-Lee weights of all vectors in $\mathbb{Z}_q^k$ is given by $\sum_{\mathbf{u}_i\in\mathbb{Z}_q^k}\mathrm{w}_p(\mathbf{u}_i)
    =
    \frac{kq^{k-1}(2q^2+1)}{6},$
    whose derivation is given in Appendix~\ref{prop:sum_pair_weight}. Denoting the sum of the pair-Lee weights of the vectors in $\ker(f)$ by $s$, we have 
     \begin{align}
    	\sum_i[\mathbf{D}]_{ij}
    	&\ge
    	(q^k-q^{k-l})\left(2t+1-\frac{q}{2}\right)
    	-\frac{kq^{k-1}(2q^2+1)}{6}+s. \nonumber
    \end{align} Thus, 
    \begin{align} \label{lb_D}
    	\sum_{i,j}[\mathbf{D}]_{ij}
    	& \ge
    	q^k\sum_i[\mathbf{D}]_{ij} \nonumber \\
    	& = q^k \Big[(q^k-q^{k-l})\left(2t+1-\frac{q}{2}\right)
    	-\frac{kq^{k-1}(2q^2+1)}{6}+s\Big]. 
    \end{align} 
    From the proof of Theorem~\ref{Plotkin_bound}, it follows that, for $M\equiv0\pmod4$,
    \begin{align} \label{ub_D}
    	\sum_{i,j}[\mathbf{D}]_{ij}
    	& \le
    	\frac {3qM^2r}{8}. 
    \end{align}
    From (\ref{lb_D}) and (\ref{ub_D}), we obtain
    \begin{align}
    	\frac{3qM^2r}{8} \ge q^k \Bigg[
    	(q^k-q^{k-l})\left(2t+1-\frac{q}{2}\right)
    	-\frac{kq^{k-1}(2q^2+1)}{6}
    	+s
    	\Bigg],
    \end{align} where $M=q^k$.
    Which implies
    	\begin{equation*}
    	r_{p}^f(q, k, t) \ge \frac {8} {3q}\Big[(2t+1-\frac q2)(1-q^{-l})-\frac {k(2q^2+1)}{6q}+\frac {s}{q^k}\Big],
    \end{equation*} where $s=\sum_{\mathbf u \in \ker(f)} \mathrm{w}_p(\mathbf u)$.
\end{IEEEproof}
As a special case, for $q=2$, the proposed bound reduces to the Plotkin-type bound for function-correcting symbol-pair codes (FCSPCs) established in \cite{SR}, given by
\begin{equation*}
	r_{p}^{f}(2, k, t)
	\ge
	\frac{8t}{3}\left(1-2^{-l}\right)
	-k
	+\frac{s}{3\cdot 2^{k-2}},
\end{equation*}
where $s=\sum_{\mathbf{u}\in\ker(f)}\mathrm{w}_p(\mathbf{u}).$

\section{Conclusion} \label{conclusion}
In this paper, we generalized symbol-pair Lee-distance from $\mathbb{Z}_4$ to $\mathbb{Z}_q$, where $q \ge 2$, thereby establishing a Lee-metric framework for $q$-ary symbol-pair read channels. Using this framework for $q=2^m$, where $m \ge 1$, we introduced function-correcting symbol-pair Lee-distance codes (FCSPLCs) and investigated their redundancy requirements. By establishing a connection between FCSPLCs and irregular-pair Lee-distance codes, we derived general lower and upper bounds on the optimal redundancy in terms of the shortest length of such codes, along with simplified bounds that are easier to evaluate. Furthermore, by deriving lower and upper bounds on the shortest length of irregular-pair Lee-distance codes, we established Plotkin-type and Gilbert--Varshamov-type bounds on the optimal redundancy of FCSPLCs. Specializing to bijective functions yielded the first Plotkin-type and Gilbert--Varshamov-type bounds for classical Lee metric codes for symbol-pair read channels over $\mathbb{Z}_q$. We further specialized the proposed framework to several important functions, including pair-locally bounded functions, the pair-Lee weight function, and the pair-Lee weight distribution function, for which we presented explicit constructions of FCSPLCs together with corresponding upper bounds on the optimal redundancy. In addition, we studied FCSPLCs associated with linear functions and derived a Plotkin-type lower bound on their optimal redundancy. Several interesting directions for future research remain, including deriving tighter bounds on the optimal redundancy of FCSPLCs and developing optimal constructions for broader classes of functions.
\appendix
	\begin{prop}\label{prop:sum_pair_weight}
		Let $q=2^m$, $m\ge1$. Then the sum of the symbol-pair Lee weights of all vectors in $\mathbb{Z}_q^k$ is
		\[
		\sum_{\boldsymbol{u}\in\mathbb{Z}_q^k}w_{p}(\boldsymbol{u})
		=
		\frac{kq^{k-1}(2q^2+1)}{6}.
		\]
	\end{prop}
	
	\begin{IEEEproof}
		Recall that the symbol-pair Lee weight of
		$\boldsymbol{u}=(u_1,\ldots,u_k)\in\mathbb{Z}_q^k$ is
		\[
		w_{p}(\boldsymbol{u})
		=
		\sum_{i=1}^{k}
		\max\!\left\{
		w_L(u_i),\,w_L(u_{i+1})
		\right\},
		\]
		where $u_{k+1}=u_1$.
		Then,
		\begin{align}
			\sum_{\boldsymbol{u}\in\mathbb{Z}_q^k}
			w_{p}(\boldsymbol{u})
			&=
			\sum_{\boldsymbol{u}\in\mathbb{Z}_q^k}
			\sum_{i=1}^{k}
			\max\!\left\{
			w_L(u_i),w_L(u_{i+1})
			\right\}
			\nonumber\\
			&=
			\sum_{i=1}^{k}
			\sum_{\boldsymbol{u}\in\mathbb{Z}_q^k}
			\max\!\left\{
			w_L(u_i),w_L(u_{i+1})
			\right\}.
			\label{eq:pair_sum1}
		\end{align}
		Fix an index $i\in[k]$. Since the remaining $k-2$ coordinates are arbitrary, every ordered pair
		$(a,b)\in\mathbb{Z}_q^2$
		appears exactly $q^{k-2}$ times as $(u_i,u_{i+1})$. Therefore,
		\[
		\sum_{\boldsymbol{u}\in\mathbb{Z}_q^k}
		\max\!\left\{
		w_L(u_i),w_L(u_{i+1})
		\right\}
		=
		q^{k-2}T(q),
		\]
		where
		\[
		T(q)
		=
		\sum_{a,b\in\mathbb{Z}_q}
		\max\!\left\{
		w_L(a),w_L(b)
		\right\}.
		\]
		
		It remains to evaluate $T(q)$. Since $q=2^m$ is even, write
		$q=2l$. The possible Lee weights are
		$0,1,\ldots,l$ with multiplicities $c_0=c_l=1,\,
		c_j=2,\,
		1\le j\le l-1.$
		
		For $0\le t\le l$, let $A_t
		=
		\#\{a\in\mathbb{Z}_q:\,w_L(a)\le t\}.$
		Then
		\[
		A_t=
		\begin{cases}
			2t+1,&0\le t\le l-1,\\
			2l,&t=l.
		\end{cases}
		\]
		
		Let $N_t
		=
		\#\{(a,b)\in\mathbb{Z}_q^2:
		\max(w_L(a),w_L(b))=t\}.$
		Since $N_t=A_t^2-A_{t-1}^2,$
		we obtain
		\[
		N_t=
		\begin{cases}
			8t,&1\le t\le l-1,\\
			4l-1,&t=l.
		\end{cases}
		\]
		Therefore,
		\begin{align}
			T(q)
			&=
			\sum_{t=1}^{l}tN_t
			\nonumber\\
			&=
			8\sum_{t=1}^{l-1}t^2
			+
			l(4l-1)
			\nonumber\\
			&=
			8\cdot
			\frac{(l-1)l(2l-1)}{6}
			+
			l(4l-1)
			\nonumber\\
			&=
			\frac{8l^3+l}{3}
			=
			\frac{q(2q^2+1)}{6}.
			\label{eq:Tq}
		\end{align}
		
		Substituting \eqref{eq:Tq} into \eqref{eq:pair_sum1} yields
		\[
		\sum_{\boldsymbol{u}\in\mathbb{Z}_q^k}
		w_{p}(\boldsymbol{u})
		=
		kq^{k-2}
		\cdot
		\frac{q(2q^2+1)}{6},
		\]
		which simplifies to
		\[
		\sum_{\boldsymbol{u}\in\mathbb{Z}_q^k}
		w_{p}(\boldsymbol{u})
		=
		\frac{kq^{k-1}(2q^2+1)}{6}.
		\]
	\end{IEEEproof}


\begin{thebibliography}{1}
  \bibitem{Lenz}
  A. Lenz, R. Bitar, A. Wachter-Zeh, and E. Yaakobi, ``Function-correcting codes,'' IEEE Transactions on Information Theory, vol. 69, no. 9, pp. 5604-5618, 2023.	
  	
  \bibitem{Cassuto}
  Y. Cassuto and M. Blaum, ``Codes for symbol-pair read channels,'' IEEE
  Trans. Inf. Theory, vol. 57, no. 12, pp. 8011-8020, Dec. 2011.
  
  \bibitem{XLC}
  Q. Xia, H. Liu and B. Chen, ``Function-Correcting Codes for Symbol-Pair Read Channels,'' in IEEE Transactions on Information Theory, vol. 70, no. 11, pp. 7807-7819, Nov. 2024.
  
  \bibitem{PAA}
  Kumar, Pratyush, Ankur Singh, and Abhay Kumar Singh. ``Lee metric codes for symbol-pair read channels,'' Discrete Mathematics 349.3, 114841, 2026.
  
  \bibitem{CW} 
  J. C.Y. Chiang and J. K. Wolf, ``On channels and codes for the Lee metric,'' Information and Control, vol. 19, no. 2, pp. 159–173, 1971.
  
  \bibitem{KN}
  K. Nakamura, ``A class of error correcting codes for DPSK channels,'' in Proceedings of the International Conference on Communications (ICC '79), vol. 3, p. 45, 1979.
  
  \bibitem{AB}
  A. Barg and A. Mazumdar, ``Codes in permutations and error correction for rank modulation,'' IEEE Transactions on Information Theory, vol. 56, no. 7, pp. 3158-3165, 2010.
  
  \bibitem{GA}
  Gyanendra K. Verma, Abhay Kumar Singh, ``On Function-Correcting Codes in the Lee Metric,''  Available on arXiv:2507.17654 [cs.IT], July 2025. 
  
  \bibitem{HRS}
  K. Hareesh, N. T. Rashid Ummer and B. Sundar Rajan, ``Plotkin-like Bound and Explicit Function-Correcting Code Constructions for Lee Metric Channels'', IEEE Transactions on Information Theory, vol. $72$, no.$8$, pp. $5639-5657$, Aug. 2026.
  
  \bibitem{PR}
  R. Premlal and B. Sundar Rajan, ``On function-correcting codes,'' in IEEE Transactions on Information Theory, vol. 71, no. 8, pp. 5884-5897, Aug. 2025.
  
  \bibitem{LS}
  H. Ly and E. Soljanin, ``On the Redundancy of Function-Correcting Codes over Finite Fields,'' 2025 13th International Symposium on Topics in Coding (ISTC), Los Angeles, CA, USA, 2025, pp. 1-5.
  
  \bibitem{ZXZG}
  Y. Zhang, Z. Xu, X. Zhang and G. Ge, ``Optimal Redundancy of Function-Correcting Codes,'' in IEEE Transactions on Information Theory, vol. 71, no. 12, pp. 9458-9467, Dec. 2025.
  
  \bibitem{CSFH}
  C. Rajput, B. S. Rajan, R. Freij-Hollanti and C. Hollanti, ``Function-Correcting Codes for Locally Bounded Functions,'' 2025 IEEE Information Theory Workshop (ITW), Sydney, Australia, 2025, pp. 851-856.
  
  \bibitem{CSFH2}
  C. Rajput, B. S. Rajan, R. Freij-Hollanti, and C. Hollanti. ``Function-Correcting Partition codes,'' Available on arXiv:2601.06450 [cs.IT], Jan. 2026.
  
  \bibitem{CSFH3}
  C. Rajput, B. Sundar Rajan, R. Freij-Hollanti and C. Hollanti, ``Function-Correcting Codes With Data Protection,'' IEEE Transactions on Information Theory, vol. $72$, no.$7$, pp. $4860-4880$, July 2026.
  
  \bibitem{SAARB}
  Swaraj Sharma Durgi,  Anjana A. Mahesh, Anupriya Kumari, Rajlaxmi Pandey, and B. S. Rajan, ``Function-correcting codes with optimal data protection for Hamming code membership,'' Available on arXiv:2602.21932 [cs.IT], Feb. 2026.
  
  \bibitem{AAS}
  Adityawardhan Yadava, Anjana A. Mahesh, and Swaraj Sharma Durgi, ``Function-Correction with Optimal Data Protection for the General Hamming Code Membership.'' Available on arXiv:2605.14023 [cs.IT], May 2026. 
  
  \bibitem{ASAK}
  Anamika Singh, and Abhay Kumar Singh, ``Data Protection in Function-Correcting Symbol-Pair Codes: Redundancy Bounds and
  Protection Profiles,'' Available on arXiv:2609.10989 [cs.IT], September 2026.
  
  
  \bibitem{AAY}
  Anamika Singh, Abhay Kumar Singh, and Eitan Yaakobi, ``Function-correcting codes for b-symbol read channels,'' Available on arXiv:2503.12894 [cs.IT], March 2025.
  
  \bibitem{SR}
  Sachin S. H and B. S. Rajan, ``On Plotkin Bound for Function-Correcting Codes for b-Symbol Read Channels,'' 2025 IEEE Information Theory Workshop (ITW), Sydney, Australia, 2025, pp. 698-703.
  
  \bibitem{LL}
  Huiying Liu, Hongwei Liu, ``Function-correcting codes with homogeneous distance,'' Finite Fields and Their Applications, Volume 112, 2026.
  
  
  \bibitem{AA}
  Anamika Singh and Abhay Kumar Singh, ``Function-Correcting Codes for Insertion-Deletion Channel,'' Available on arXiv:2512.07243 [cs.IT], Dec. 2025.
  
  \bibitem{FCBSC}
  G. K. Verma, A. Singh and A. Kumar Singh, ``Function-correcting b-symbol codes for locally $(\lambda, \rho, b)$-functions,'' IEEE Transactions on Information Theory, vol. 72, no. 1, pp. 331-341, Jan. 2026.
\end{thebibliography}
\end{document}